\documentclass[submission,copyright,creativecommons]{eptcs}
\pdfoutput=1
\providecommand{\event}{ACT 2019} % Name of the event you are submitting to
\usepackage{underscore}           % Only needed if you use pdflatex.

\usepackage{amsmath,amssymb,amsthm}
\usepackage{bbold}
\usepackage{mathtools}
\usepackage{etoolbox}
\usepackage{xifthen}
\usepackage{xcolor}

\usepackage{array}
\usepackage{graphicx}
\usepackage{subcaption}
\usepackage{pdflscape}
\usepackage{afterpage}

\usepackage{thmtools}

\usepackage{thm-restate}

\newtheorem{theorem}{Theorem}
\newtheorem{lemma}{Lemma}
\newtheorem{definition}{Definition}

\newtheorem{corollary}{Corollary}
\newtheorem{example}{Example}
\newtheorem{assumption}{Assumption}

\newenvironment{romanenumerate}{
\begingroup
\renewcommand\labelenumi{(\roman{enumi})}
\renewcommand\theenumi\labelenumi
\begin{enumerate}}{\end{enumerate}
\endgroup}

\colorlet{grey}{black!20!white}

\colorlet{h1color}{blue!40!black} % highlight color 1
\colorlet{h2color}{orange!90!black} % highlight color 1
\colorlet{h3color}{blue!40!white} % highlight color 1
\colorlet{h4color}{green!40!black} % highlight color 1

\newcommand{\ddInline}[1]{\xLeftarrow {\text{\tiny\;\; $#1$}\;}}

 \newcommand{\inputtikz}[1]{%
  \vcenter{\hbox{\includegraphics{diagrams/#1.pdf}}}%
 }

 \newcommand{\inputtikzNoVcenter}[1]{%
  \includegraphics{diagrams/#1.pdf}%
 }

\makeatletter
\def\bstr{b}
\def\bfstr{bf}
\def\cstr{c}
\def\fstr{f}
\def\lst{A,B,C,D,d,E,F,G,H,I,J,K,L,M,N,O,P,Q,R,S,T,U,V,W,X,Y,Z}

\newcommand{\MkB}[1]{\expandafter\def\csname\bstr#1\endcsname{\mathbb{#1}}}
  \@for\i:=\lst\do{%
    \expandafter\MkB \i     }
    
\newcommand{\MkBF}[1]{\expandafter\def\csname\bfstr#1\endcsname{\mathbf{#1}}}
  \@for\i:=\lst\do{%
    \expandafter\MkBF \i     }

\newcommand{\MkCal}[1]{\expandafter\def\csname\cstr#1\endcsname{\mathcal{#1}}}
  \@for\i:=\lst\do{%
    \expandafter\MkCal \i     }
    
\newcommand{\MkFrak}[1]{\expandafter\def\csname\fstr#1\endcsname{\mathfrak{#1}}}
  \@for\i:=\lst\do{%
    \expandafter\MkFrak \i     }
    
\makeatother

\newcommand{\pB}[1]{\mathsf{PB}(#1)}
\newcommand{\pO}[1]{\mathsf{PO}(#1)}

\newcommand{\mono}[1]{\mathsf{mono}(#1)}

\newcommand{\epi}[1]{\mathsf{epi}(#1)}

\newcommand{\mor}[1]{\mathsf{mor}(#1)}

\newcommand{\iso}[1]{\mathsf{iso}(#1)}

\newcommand{\obj}[1]{\mathsf{obj}(#1)}

\newcommand{\Lin}[1]{\mathsf{Lin}(#1)}

\newcommand{\LinAc}[1]{\overline{\mathsf{Lin}}(#1)}

\newcommand{\MatchGT}[3]{\mathsf{M}^{{\text{\tiny $#1$}}}_{#2}(#3)}

\newcommand{\tMatchGT}[3]{\mathsf{MT}^{#3}_{#1}(#2)}

\newcommand{\compGT}[4]{#2 {}^{#3}\!{\triangleleft}_{#1} #4}

\newcommand{\equivAbs}{\equiv_A}
\newcommand{\equivShift}{\equiv_S}
\newcommand{\TcompGT}[4]{#1{}^{#2}\!\!{\angle}_{#4} #3}

\newcommand{\ac}[1]{\mathsf{#1}} %-- \mathsf version

\newcommand{\Shift}{\mathsf{Shift}}

\newcommand{\Trans}{\mathsf{Trans}}

\newcommand{\mIO}{\varnothing}

\usepackage[ruled,vlined]{algorithm2e}

\title{Tracelets and Tracelet Analysis\\ Of Compositional Rewriting Systems}

\author{Nicolas Behr
\institute{%
	Universit\'{e} de Paris, IRIF, CNRS, F-75013 Paris, France
	\thanks{This project has received funding from the European Union's Horizon 2020 research and innovation programme under the Marie Sk\l{}odowska-Curie grant agreement No~753750.}
}
\email{nicolas.behr@irif.fr}
}

\def\titlerunning{Tracelets and Tracelet Analysis}
\def\authorrunning{N. Behr}
\begin{document}
\maketitle
%%%%%%%%%%%%%%%% 
\begin{abstract}
Taking advantage of a recently discovered associativity property of rule compositions, we extend the classical concurrency theory for rewriting systems over adhesive categories. We introduce the notion of tracelets, which are defined as minimal derivation traces that universally encode sequential compositions of rewriting rules. Tracelets are compositional, capture the causality of equivalence classes of traditional derivation traces, and intrinsically suggest a clean mathematical framework for the definition of various notions of abstractions of traces. We illustrate these features by introducing a first prototype for a framework of tracelet analysis, which as a key application permits to formulate a first-of-its-kind algorithm for the static generation of minimal derivation traces with prescribed terminal events.
\end{abstract}
%%%%%%%%%%%%%%%% 

\section{Motivation and relation to previous works}\label{sec:intro}

The analysis of realistic models of complex chemical reaction systems in organic chemistry and in systems biology poses considerable challenges, both in theory and in terms of algorithmic implementations. %
Two major classes of successful approaches include \emph{chemical graph rewriting}%
~\cite{Benk2003,banzhafetalDR20154968,Andersen2016,andersen2018rule}, and the %
\emph{rule-based modeling frameworks} \textsc{Kappa}~\cite{Danos2003aa,Danos2004aa,Danos2003ab,Danosaa,Danosab} and \textsc{BioNetGen}~\cite{Harris2016aa,Blinov2004aa}, respectively. %
These approaches  utilize well-established modern variants of \emph{Double-Pushout (DPO)}~\cite{ehrig1999handbook,ehrig2014mathcal} and \emph{Sesqui-Pushout (SqPO)} \cite{Corradini2006} rewriting frameworks over suitably chosen adhesive categories~\cite{lack2005adhesive} (and with additional constraints~\cite{habel2009correctness,ehrig2014mathcal} on objects and transitions for consistency). The sheer complexity of the spaces of distinct classes of objects and of active transitions thereof necessitated the development of specialized and highly optimized variants of static analysis techniques for these types of systems. As we will demonstrate in this paper, a novel class of such techniques is found to arise from a refocusing of the analysis from \emph{derivation traces} to so-called \emph{tracelets}. 

To provide some context, we briefly recall some basic notions of rewriting based upon a finitary adhesive category $\bfC$~\cite{lack2005adhesive}, such as e.g.\ the category $\mathbf{FinGraph}$ of finite directed multigraphs. Objects of this category provide the possible configurations or states of the rewriting system (typically considered up to isomorphism), while partial maps between objects (encoded as spans of monomorphisms) will provide the possible transitions, referred to as \emph{(linear) rules}. The application of a rule $O\xleftharpoonup{r}I$ to some object $X$ then requires the choice of an instance of a subobject $I$ within $X$, established via a monomorphism $m:I\hookrightarrow X$ called a \emph{match}, followed by replacing $m(I)\subset X$ with an instance $m^{*}(O)$ of $O$, where the precise details depend on the chosen rewriting semantics (i.e.\ \emph{Double-Pushout (DPO)}~\cite{ehrig1999handbook,ehrig2014mathcal} or \emph{Sesqui-Pushout (SqPO)}~\cite{Corradini2006} semantics). This process of rule application is traditionally referred to as a \emph{(direct) derivation}. The central structure studied in the concurrency theory and static analysis of the rewriting system consists in so-called \emph{derivation traces}: 
\begin{equation}\label{eq:lengthNDerTrace}
\inputtikz{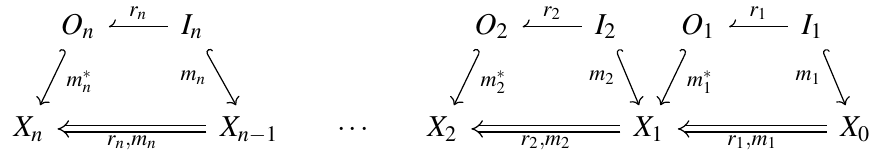}
\end{equation}
Each transition in such a derivation trace from a state $X_i$ to a state $X_{i+1}$ is thus given by a direct derivation via a linear rule $r_i$ at a match $m_i$. A typical abstract encoding of rewriting systems is then provided in the form of a \emph{rewriting grammar}, whose data consists of an initial state $X_0$ and a set of linear rewriting rules, from which all possible derivation traces starting at $X_0$ are constructed. 

Static analysis of rewriting systems is traditionally based upon several notions of \emph{abstractions} of derivation traces. At a fundamental level, the category-theoretical definitions of rewriting are inherently invariant under various types of isomorphisms, which suggests a form of equivalence on derivation traces induced by isomorphisms referred to as \emph{abstraction equivalence}~\cite{Corradini1994}. The second major source of equivalences is based upon so-called \emph{sequential independence} of derivations~\cite{ehrig1999handbook,ehrig2014mathcal,Corradini2006,BALDAN2014}: again leaving technicalities aside, if two ``adjacent'' direct derivations $X_{i+1}\ddInline{r_{i+1},m_{i+1}}X_i$ and $X_i\ddInline{r_i,m_i}X_{i-1}$ in a given derivation trace are sequentially independent, there exist matches $m_i'$ and $m_{i+1}'$ so that 
\[
X_{i+1}\ddInline{r_{i},m_{i'}}X_i \quad \text{and}\quad X_i\ddInline{r_{i+1},m_{i+1}'}X_{i-1}
\]
constitute sequential derivations in the opposite order of application. Lifting this notion to sequences of an arbitrary finite number of consecutive derivations yields an abstraction equivalence called \emph{shift equivalence}~\cite{Kreowski1987,ehrig1999handbook,ehrig2006fund}. Quotienting a given grammar by a combination of abstraction and shift equivalence leads to the sophisticated frameworks of occurrence grammars~\cite{BALDAN1999aa,baldan2000modelling} as well as (equivalently~\cite{Baldan2000,Baldan2007aa,Baldan2009aa}) of processes and unfoldings~\cite{Baldan1998aa,Baldan2006aa,BALDAN2014}. Quintessentially, since sequential commutativity induces a preorder on derivations of a grammar, the aforementioned well-established static analysis techniques encode the causal relationships of derivations according to this preorder. 

Of particular interest in view of practical applications of such techniques to chemical and biochemical reaction systems (via \emph{chemical graph rewriting}~\cite{Benk2003,banzhafetalDR20154968,Andersen2016,andersen2018rule}, and via the \emph{rule-based modeling} frameworks \textsc{Kappa}~\cite{Danos2003aa,Danos2004aa,Danos2003ab,Danosaa,Danosab} and \textsc{BioNetGen}~\cite{Harris2016aa,Blinov2004aa}) are concepts that permit to extract high-level information on the causal properties of the typically immensely complex transition sets and state spaces encountered in real-life reaction systems in an automated fashion. In the setting of systems chemistry, taking full advantage of the highly constrained type of rewriting relevant to model molecules and possible reactions (i.e.\ a flavor of DPO rewriting in which vertices modeling atoms are preserved throughout transitions), a highly efficient analysis technique based upon mapping of reaction networks into multi-hypergraphs and modeling pathways as integer hyperflows has been developed in~\cite{Fagerberg2018,andersen2018towards,Andersen2019}. An essential role in this framework is played by compositions of chemical graph rewriting rules~\cite{Andersen2013,Andersen2014,andersen2018rule}, which have been implemented algorithmically in~\cite{Andersen2016}. The tracelets as introduced in this paper may be seen as a formalization of these ideas of understanding pathways as particular rule compositions, which in particular answers an open question on the associativity of compositions of such pathways to the affirmative.

In the biochemistry setting, important developments include sophisticated specializations of the aforementioned static analysis techniques for general rewriting systems to the relevant setting of site-graph rewriting in order to extract information on cellular signaling pathways~\cite{Danosaa,Danosab,Danosac}, the notion of refinements~\cite{danos2008rule}, techniques of model reduction based on the differential semantics of the stochastic transition systems~\cite{danos2010abstracting} and notions of trace compression~\cite{danos2012graphs}. In particular, so-called \emph{strong compression} as introduced in~\cite{danos2012graphs} will play an interesting role also in our tracelet framework. While the theory of static analysis of such complex rewriting systems is thus rather well-developed, several open problems remain. Referring to~\cite{Boutillier2018aa} for a recent review, at present the established approach to the generation of pathways for biochemical reaction systems passes through extensive simulation runs in order to generate large ensembles of derivation traces of the given system, which then have to be curated and suitably compressed in order to extract the static information constituting the pathways of interest. This dependence on a posteriori analyses of derivation traces hinders the efficiency of the algorithms considerably, since typically only a small portion of the information contained in a given trace gives rise to useful information on pathways. We will develop in the following an alternative approach to the static analysis of rewriting systems that aims to avoid precisely this bottleneck in the synthesis of pathways.

The \textbf{main contribution} of this paper consists in an alternative paradigm for the static analysis of rewriting systems, which emphasizes the notion of sequential rule compositions over that of derivation traces. Our development hinges on two central theorems of rewriting theory: a theorem describing the relationship between two-step sequences of direct derivations and the underlying rule compositions, the so-called \emph{concurrency theorem} (well-known in the DPO setting~\cite{ehrig1999handbook,ehrig2006fund,ehrig2014mathcal}, only recently established in the SqPO setting~\cite{nbSqPO2019,bk2019a}), and an equally recently proved~\cite{bdg2016,bp2018,nbSqPO2019,bk2019a} theorem establishing a form of \emph{associativity} of the operation of rule compositions. The combination of these two results admits to characterize derivation traces \emph{universally} via so-called \emph{tracelets}, in the sense that each trace of length $n$ applied to an initial object $X_0$ may be obtained as the extension of a \emph{minimal} derivation trace of length $n$ into the context of the object $X_0$. Referring to Figure~\ref{fig:traceletoverview} for an overview, one may shift focus onto the tracelets themselves as the objects to analyze in a given rewriting system, since they encode all relevant information in terms of the causality of derivation traces. From a technical perspective, since a tracelet is nothing but a special type of derivation trace, all of the traditional analysis techniques on derivation traces remain applicable. At the same time, tracelets may be naturally equipped with a notion of associative composition, which opens novel possibilities in view of static pathway generation in the aforementioned (bio-) chemical rewriting system settings.

\paragraph{Plan of the paper:} In Section~\ref{sec:tSetup}, the core tracelet formalism is established, providing the precise definitions of the concepts summarized in the schematic Figure~\ref{fig:traceletoverview}. Section~\ref{sec:ta} is devoted to developing \emph{tracelet analysis}, based in part upon the aforementioned static analysis techniques for derivation traces. As a first application of our framework, we present a prototypical \emph{Feature-driven Explanatory Tracelet Analysis (FETA)} algorithm in Section~\ref{sec:FETA}. Since our framework is heavily based upon our very recent developments in the field of compositional rewriting, we provide a technical appendix containing a collection of illustrative figures and of requisite technical definitions and results.

\section{Tracelets for compositional rewriting theories}
\label{sec:tSetup}

\begin{assumption}\label{as:main}
Throughout this paper, we fix\footnote{Although especially in the DPO-type rewriting case more general settings would be admissible while retaining compositionality of the rewriting (see~\cite{bk2019a} for further details), the present choice covers many cases of interest, is a sufficient setting also for compositional Sesqui-Pushout (SqPO) rewriting, and overall strikes a good balance of generality vs.\ simplicity.} a category $\bfC$ that satisfies:
\begin{itemize}
\item $\bfC$ is \textbf{adhesive}~\cite{lack2005adhesive}
\item $\bfC$ possesses an \textbf{epi-mono-factorization}~\cite{habel2009correctness} (i.e.\ every morphism $f\in\mor{\bfC}$ can be factorized into the form $f=m\circ e$, with $m\in \mono{\bfC}$ and $e\in \epi{\bfC}$)
\item $\bfC$ possesses a \textbf{strict initial object} $\mIO\in \obj{\bfC}$~\cite{lack2005adhesive} (i.e.\ an object such that for every $X\in \obj{\bfC}$, there exists a unique monomorphism $\mIO\rightarrow X$, and for every $Y\in \obj{\bfC}$, if there exists a morphism $Z\rightarrow \mIO$, then it is an isomorphism).
\item $\bfC$ is \textbf{finitary}, i.e.\ for every object $X\in \obj{\bfC}$, there exist only finitely many monomorphisms $Z\rightarrow X$ into $X$ (and thus only finitely many subobjects of $X$).
\end{itemize}
\end{assumption}
Categories satisfying Assumption~\ref{as:main} have a number of properties that are of particular importance in view of compositionality of rewriting rules (cf.\ Appendix~\ref{app:ACprops}). A prototypical example of a category satisfying all of the assumptions above is the finitary restriction $\mathbf{FinGraph}$ of the category $\mathbf{Graph}$ of \emph{directed multigraphs}~\cite{Braatz2010aa}. We collect in Appendix~\ref{app:CR} the necessary background material on compositional DPO- and SqPO-type rewriting for rules with conditions~\cite{bdg2016,bp2018,nbSqPO2019,bk2019a}, and will freely employ the standard notations therein.

\begin{definition}[Tracelets]\label{def:genT}
	Let $\bT\in\{DPO,SqPO\}$ be the type of rewriting, and let $\LinAc{\bfC}$ denote the set of linear rules with conditions over $\bfC$ (cf.\ Definition~\ref{def:RwCs}). 
	\begin{itemize}
		\item \textbf{Tracelets of length $1$:} the set $\cT^{\bT}_1$ of type $\bT$ tracelets $T(R)$ of length $1$ is defined as
\begin{equation}
	\cT^{\bT}_1:=\left.\left\{
					T(R)=\inputtikz{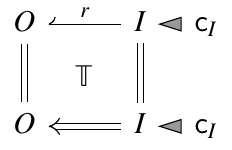}
				\right\vert
				R=(r,\ac{c}_I)\in \LinAc{\bfC}
			\right\}\,.
		\end{equation}
	\item \textbf{Tracelets of length $n+1$:} given tracelets $T_{n+1}\in \cT^{\bT}_1$ of length $1$ and $T_{n\cdots 1}\in \cT^{\bT}_n$ of length $n$ (for $n\geq1$), we define a span of monomorphisms ${\color{h1color}\mu}={\color{h1color}(I_{n+1}\hookleftarrow M\hookrightarrow O_{n\cdots 1})}$ as \emph{$\bT$-admissible match}, denoted $\mu\in \tMatchGT{T_1}{T_{n\cdots 1}}{\bT}$, if the following diagram is constructable:
	\begin{equation}\label{eq:tlGenA}
		\inputtikz{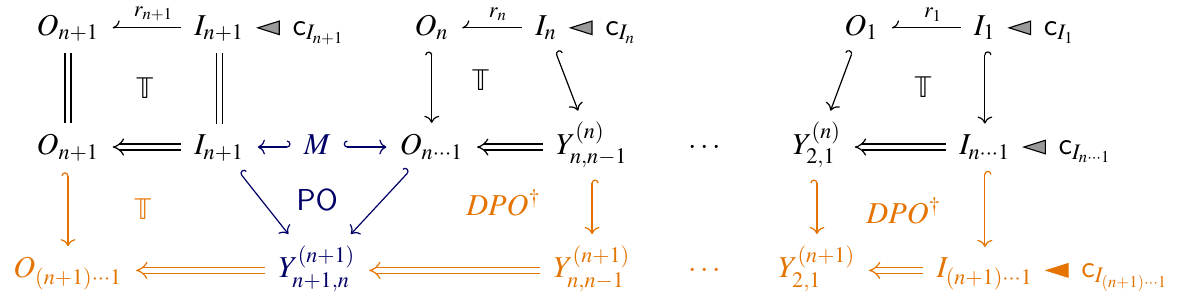}
%		\vcenter{\hbox{\includegraphics[scale=0.8]{images/tlGen.pdf}}}
	\end{equation}
	Here, the square marked ${\color{h1color}\mathsf{PO}}$ is constructed as a pushout, followed by performing the ${\color{h2color}\bT}$- and ${\color{h2color}DPO^{\dag}}$-type direct derivations as indicated to form the lower part of the diagram. The latter operation may fail, either by non-existence of the requisite pushout complements (cf.\ Definition~\ref{def:rew}), or, if all POCs exist, because the tentative composite condition ${\color{h2color}\ac{c}_{I_{(n+1)\cdots 1}}}$ might evaluate to $\ac{false}$, with ${\color{h2color}\ac{c}_{I_{(n+1)\cdots 1}}}$ computed as 
	\begin{equation}
	\begin{aligned}
		\ac{c}_{I_{(n+1)\cdots}}&:=\Shift(I_{n\cdots 1}\hookrightarrow {\color{h2color}I_{(n+1)\cdots 1}},\ac{c}_{I_{n\cdots 1}})\\
		&\quad \bigwedge \;
		\Trans({\color{h1color}Y^{(n+1)}_{n+1,n}}\Leftarrow {\color{h2color}I_{(n+1)\cdots 1}},\Shift(I_{n+1}\hookrightarrow {\color{h1color}Y^{(n+1)}_{n+1,n}},\ac{c}_{I_{n+1}}))\,.
	\end{aligned}
	\end{equation}
	If ${\color{h1color}\mu}\in \tMatchGT{T_1}{T_{n\cdots 1}}{\bT}$, we define the tracelet $\TcompGT{T_{n+1}}{{\color{h1color}\mu}}{T_{n\cdots 1}}{\bT}$ of length $n+1$ as
	\begin{equation}\label{eq:tlGenB}
		\TcompGT{T_{n+1}}{{\color{h1color}\mu}}{T_{n\cdots 1}}{\bT}:=
		\inputtikz{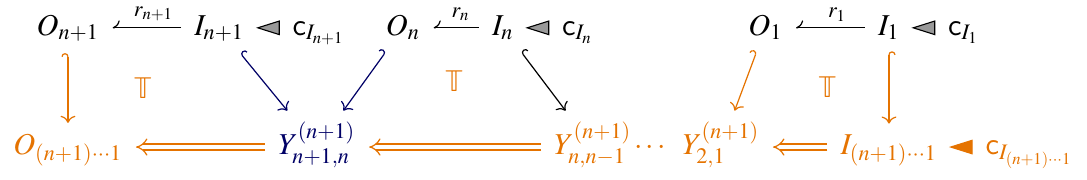}
	\end{equation}
	We define the set $\cT^{\bT}_{n+1}$ of \emph{type $\bT$ tracelets of length $n+1$} as
	\begin{equation}
		\cT^{\bT}_{n+1}:=\left.\left\{
		\TcompGT{T_{n+1}}{{\color{h1color}\mu}}{T_{n\cdots 1}}{\bT}
		\right\vert 
			T_{n+1}\in \cT^{\bT}_1\,,\;
			T_{n\cdots 1}\in \cT^{\bT}_{n}\,,\;
			\mu\in \tMatchGT{T_1}{T_{n\cdots 1}}{\bT}
		\right\}\,.
	\end{equation}
	\end{itemize}
For later convenience, we introduce the \emph{tracelet evaluation operation} $[[.]]$, 
\begin{equation}\label{eq:Tev}
	[[.]]:\cT^{\bT}\rightarrow \LinAc{\bfC}: \cT^{\bT}_n\ni T \mapsto[[T]]:=((O_{n\cdots 1}\leftharpoonup I_{n\cdots 1}),\ac{c}_{I_{n\cdots 1}})\,,
\end{equation}
with $\cT^{\bT}:=\bigcup_{n\geq 1}\cT^{\bT}_n$, and where $(O_{n\cdots 1}\leftharpoonup I_{n\cdots 1})$ denotes the span composition (cf.~\eqref{eq:tlGenA})
\begin{equation}
(O_{n\cdots 1}\leftharpoonup I_{n\cdots 1}):=
(O_{n\cdots 1}\Leftarrow Y^{(n)}_{n,n-1})\circ \cdots
\circ (Y^{(n)}_{2,1}\Leftarrow I_{n\cdots 1})\,.
\end{equation}
\end{definition}

A first example of a tracelet of length $3$ generated iteratively from tracelets of length $1$ is given in Figure~\ref{fig:traceletAssociativityC}, with the relevant computation presented in (the top half of) Figure~\ref{fig:traceletAssociativityB}. The example illustrates a sequential composition of graph rewriting rules, with vertex symbols and edge colors used purely to encode the structure of the various morphisms and rules, i.e.\ repeated symbols mark objects identified by the partial morphisms. Note that since in this example no vertices are deleted without explicitly deleting the incident edges, too, this example constitutes a valid composition in both the DPO- and the SqPO-type frameworks. 

Another very important aspect visualized in Figure~\ref{fig:traceletAssociativityB} is the \emph{associativity} property of the underlying rule compositions: the top half of the figure represents a composition of $r_2$ with $r_1$ (yielding the tracelet of length $2$ highlighted in blue), followed by a further composition of $r_3$ with the composite of $r_2$ and $r_1$. By the associativity theorem for compositional rewriting theories (Theorem~\ref{thm:assocR}), there exist suitable overlaps such that the outcome of the aforementioned operation may be equivalently obtained by composing $r_3$ with $r_2$ (yielding the tracelet of length $2$ highlighted in yellow), and by pre-composing the composite with $r_1$. Vertically composing squares in each half of Figure~\ref{fig:traceletAssociativityB}, one may verify that this associativity property on rule compositions extends to an associativity property on tracelet compositions, as both halves of the figure yield the same tracelet of length $3$. These observations motivate the following extension of the definition of $\TcompGT{.}{.}{.}{\bT}$:
\begin{definition}[Tracelet composition]\label{def:tlComp}
For tracelets $T',T\in\cT^{\bT}$ of lengths $m$ and $n$, respectively, a span of monomorphisms ${\color{h1color}\mu}=(I'_{m\cdots 1}\hookleftarrow M\hookrightarrow O_{n\cdots 1})$ is defined to be an \emph{admissible match of $T$ into $T'$}, denoted $\mu\in \tMatchGT{T'}{T}{\bT}$, if (i) all requisite pushout complements exist to form the type $DPO^{\dag}$ derivations (in the sense of rules without conditions) to construct the diagram in~\eqref{eq:TcompDefA} below, where $p:=m+n+1$,
\begin{subequations}
\begin{equation}\label{eq:TcompDefA}
\inputtikz{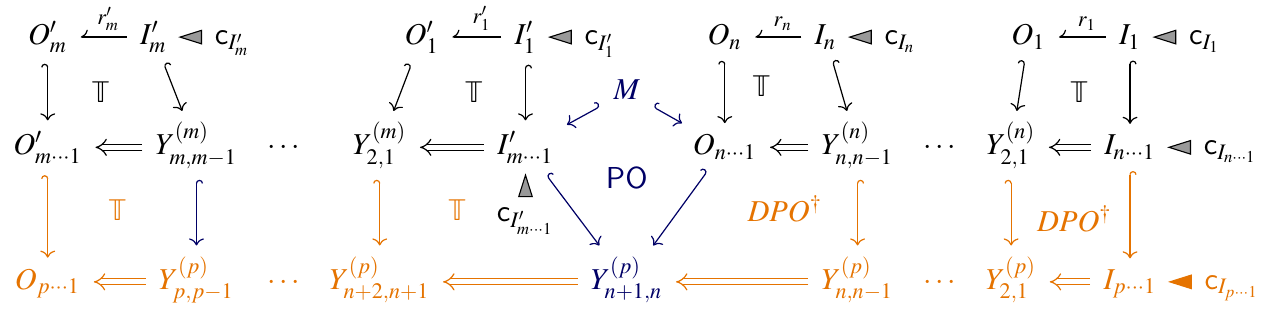}
\end{equation}
and if (ii) the condition ${\color{h2color}\ac{c}_{I_{(m+n+1)\cdots 1}}}$ as in~\eqref{eq:TcompDefB} below does not evaluate to $\ac{false}$:
\begin{equation}
\begin{split}
	{\color{h2color}\ac{c}_{I_{(m+n+1)\cdots}}}
		&:=\Shift(I_{n\cdots 1}\hookrightarrow {\color{h2color}I_{(m+n+1)\cdots 1}},\ac{c}_{I_{n\cdots 1}})\\
		&\quad \bigwedge \;
		\Trans({\color{h1color}Y^{(m+n+1)}_{n+1,n}}\Leftarrow {\color{h2color}I_{(m+n+1)\cdots 1}},\Shift(I_{m\cdots 1}\hookrightarrow {\color{h1color}Y^{(n+1)}_{n+1,n}},\ac{c}_{I_{m\cdots 1}}))\,.
\end{split}\label{eq:TcompDefB}
\end{equation}
\end{subequations}
Then for ${\color{h1color}\mu}\in \tMatchGT{T'}{T}{\bT}$, we define the \emph{type $\bT$ tracelet composition of $T'$ with $T$ along $\mu$} as
\begin{equation}\label{eq:tlCompB}
	\TcompGT{T'}{{\color{h1color}\mu}}{T}{\bT}:=
	\inputtikz{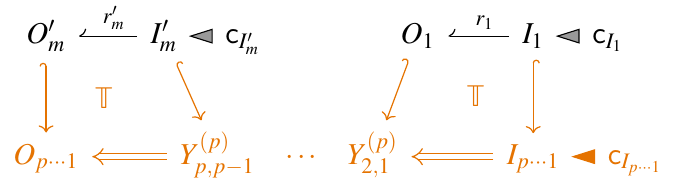}\,.
\end{equation}
\end{definition}

%%%%%%%%%%%%%%%%%%%%%%%%%

\begin{figure}
  \begin{subfigure}[t]{\textwidth}
  \centering
    \includegraphics[scale=0.35]{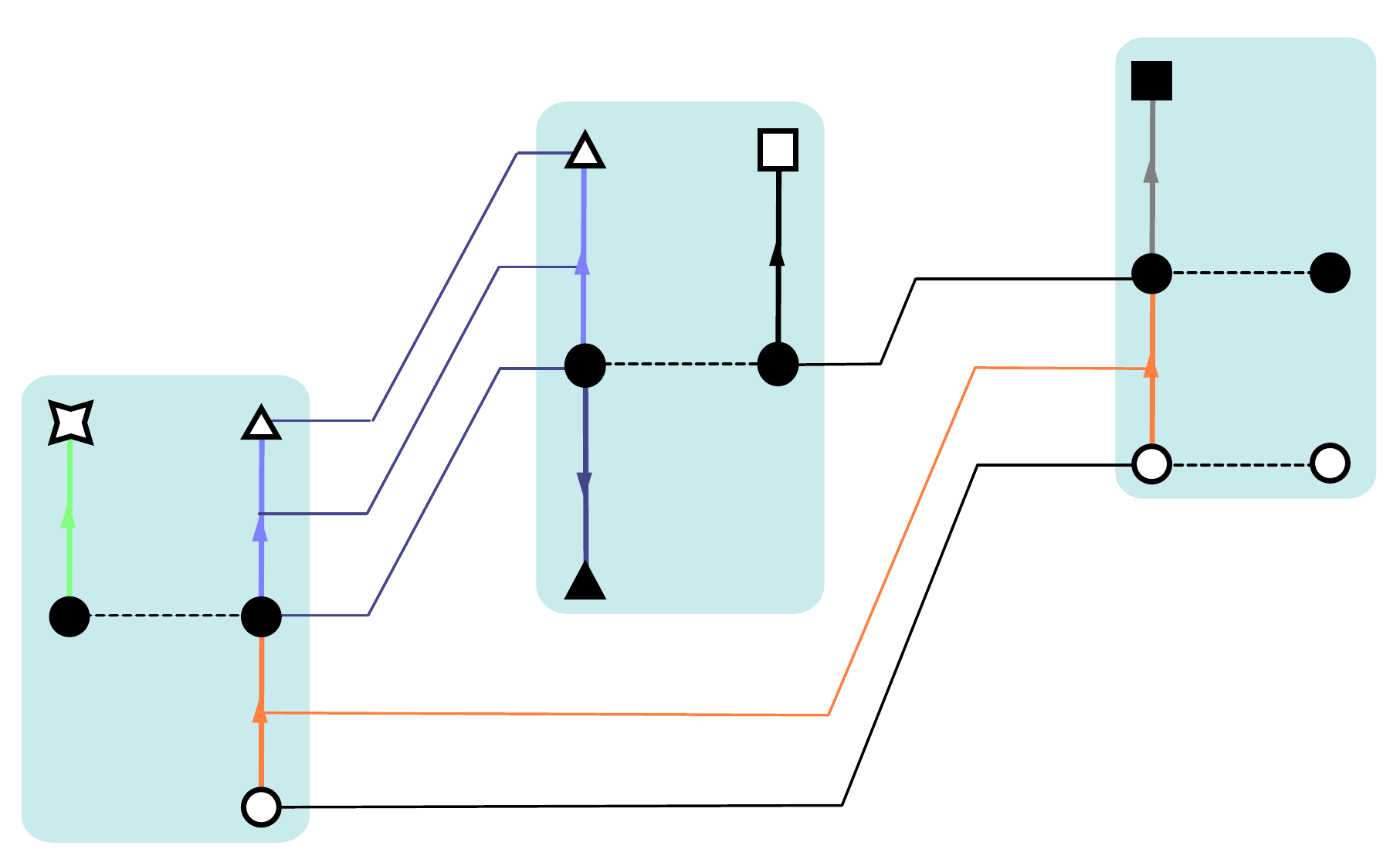}
    \caption{\label{fig:traceletAssociativityA}Three rules sequentially composed (from right to left): input and output interfaces are drawn explicitly, while the context graphs $K_j$ are implicitly encoded as subgraphs of $O_j$ and $I_j$ joined by dotted lines (for $j=1,2,3$). The structure of the matches of the rules is indicated via lines connecting elements of outputs to elements of inputs of rules.}
  \end{subfigure}\\
   \begin{subfigure}[t]{\textwidth}
   \centering
    \includegraphics[scale=0.3]{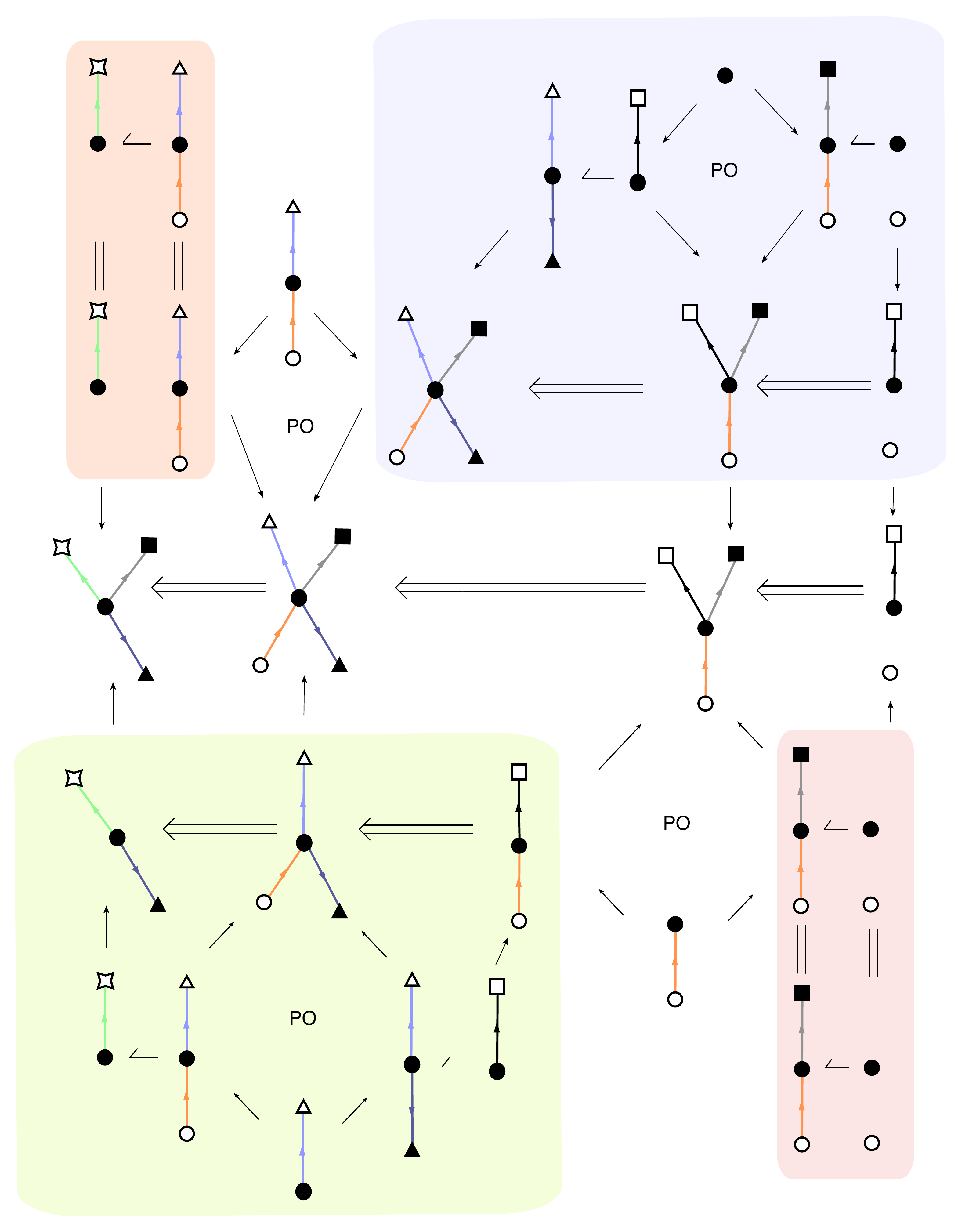}
    \caption{\label{fig:traceletAssociativityB}Explicit demonstration of the \emph{associativity property} of the rule composition operation: the top half of the diagram encodes a composition of the shape $\compGT{}{r_3}{}{(\compGT{}{r_2}{}{r_3})}$, while the bottom half encodes $\compGT{}{(\compGT{}{r_3}{}{r_2})}{}{r_1}$, with both operations for the overlaps depicted leading to the same minimal trace (up to isomorphisms). The tracelet of length $3$ equivalently encoded by both halves of the diagram is obtained by composition of squares.}
  \end{subfigure}
  \begin{subfigure}[t]{\textwidth}
  \centering
    \includegraphics[scale=0.35]{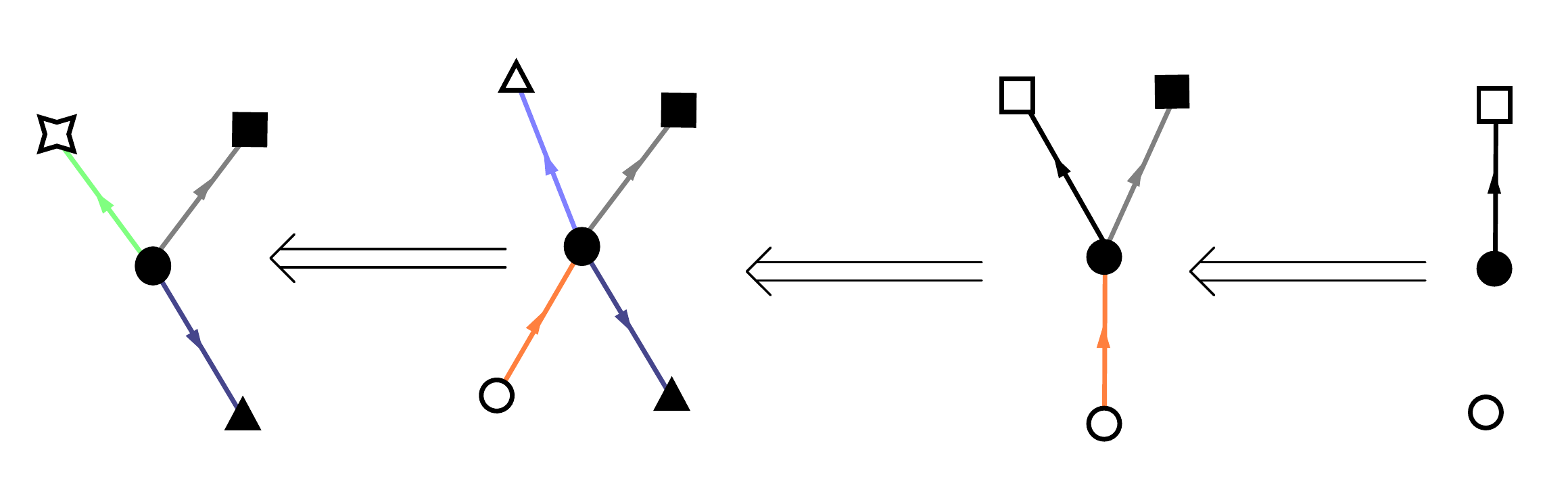}
    \caption{\label{fig:traceletAssociativityC}The minimal trace of length $3$ encoded in Fig.~\ref{fig:traceletAssociativityB}.}
  \end{subfigure}\\
  \caption{Illustration of the relationship between associativity and tracelets.}\label{fig:traceletAssociativity}
\end{figure}

%%%%%%%%%%%%%%%%%%%%%%%%

Next, the precise relationship between $\bT$-type rule and tracelet compositions is clarified.
\begin{theorem}\label{thm:TmainA}
Let $\compGT{\bT}{.}{.}{.}$ denote the $\bT$-type rule composition (Definition~\ref{def:RwCs}), and let the set of $\bT$-admissible matches be denoted by $\MatchGT{\bT}{r_2}{r_1}$ (for $r_2,r_1\in \LinAc{\bfC}$).
\begin{romanenumerate}
\item\label{thm:TmainApartI} For all $T',T\in \cT^{\bT}$, $\tMatchGT{T'}{T}{\bT}=\MatchGT{\bT}{[[T']]}{[[T]]}$.
\item\label{thm:TmainAii} For all $T',T\in\cT^{\bT}$ and $\mu\in \tMatchGT{T'}{T}{\bT}$, $\left[\left[\TcompGT{T'}{\mu}{T}{\bT}\right]\right]
  =\compGT{\bT}{[[T']]}{\mu}{[[T]]}$.
\item\label{thm:TmainAiii} The $\bT$-type tracelet composition is \emph{\textbf{associative}}, i.e.\ for any three tracelets $T_1,T_2,T_3\in \cT^{\bT}$, there exists a bijection $\varphi:S_{3(21)}\xrightarrow{\cong}S_{(32)1}$ between the sets of pairs of $\bT$-admissible matches of tracelets (with $T_{ji}:=\TcompGT{T_j}{\mu_{ji}}{T_i}{\bT}$ and using property (i))
\begin{equation}
\begin{aligned}
	S_{3(21)}&:=\{ (\mu_{21},\mu_{3(21)})\vert 
		\mu_{21}\in \MatchGT{\bT}{[[T_2]]}{[[T_1]]}\,,\; 
		\mu_{3(21)}\in  \MatchGT{\bT}{[[T_3]]}{[[T_{21}]]}\\
	S_{(32)1}&:=\{ (\mu_{32},\mu_{(32)1})\vert 
		\mu_{32}\in \MatchGT{\bT}{[[T_3]]}{[[T_2]]}\,,\; 
		\mu_{(32)1}\in \MatchGT{\bT}{[[T_{32}]]}{[[T_1]]}\}
\end{aligned}
\end{equation} 
such that for all $(\mu_{32}',\mu_{(32)1}')=\varphi((\mu_{21},\mu_{3(21)}))$
\begin{equation}
\TcompGT{T_3}{\mu_{3(21)}}{\left(\TcompGT{T_2}{\mu_{21}}{T_1}{\bT}\right)}{\bT}
\cong \TcompGT{\left(\TcompGT{T_3}{\mu'_{32}}{T_2}{\bT}\right)}{\mu_{(32)1}'}{T_3}{\bT}\,.
\end{equation}
Moreover, the bijection $\varphi$ coincides with the corresponding bijection provided in the associativity theorem for $\bT$-type rule compositions (Theorem~\ref{thm:assocR}). (\textbf{Proof:} Appendix~\ref{app:TmainAproof})
\end{romanenumerate}
\end{theorem}

Finally, combining the associativity results for rule and tracelet compositions with the so-called concurrency theorems for compositional rewriting theories, we find a characterization of derivation traces via tracelets and vice versa:
\begin{theorem}[Tracelet characterization]\label{thm:tlChar}
	For all type-$\bT$ tracelets $T\in \cT^{\bT}_n$ of length $n$, for all objects $X_0$ of $\bfC$, and for all monomorphisms $(m:I_{n\cdots 1}\hookrightarrow X_0)$ such that $m\in \MatchGT{\bT}{[[T]]}{X_0}$, there exists a type-$\bT$ direct derivation $D=T_m(X_0)$ as depicted below right obtained via vertically composing the squares in each column of the diagram below left:
\begin{equation}\label{eq:tlCharAux}
	\inputtikz{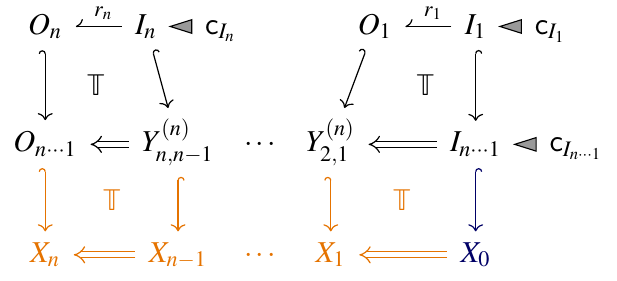}\; \leftrightsquigarrow\;
\inputtikz{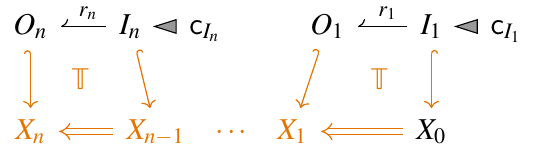}
\end{equation}
Conversely, every $\bT$-direct derivation $D$ of length $n$ along rules $R_j=(r_j,\ac{c}_{I_j})\in \LinAc{\bfC}$ starting at an object $X_0$ of $\bfC$ may be cast into the form $D=T_m(X_0)$ for some tracelet $T$ of length $n$ and a $\bT$-admissible match $m\in \MatchGT{\bT}{[[T]]}{X_0}$ that are uniquely determined from $D$ (up to isomorphisms). (\textbf{Proof:} Appendix~\ref{app:tlChar})
\end{theorem}

\section{Tracelet Analysis}
\label{sec:ta}

Let us first introduce a convenient shorthand notation for tracelets, which emphasizes that by definition, every tracelet is a type of commutative diagram of consecutive direct derivations ``glued'' at common interface objects (compare Figure~\ref{fig:traceletoverviewA}):

\begin{definition}
	For a tracelet $T\in \cT^{\bT}_n$ of length $n\geq 1$, let symbols $t_j$ for $1\leq j\leq n$ denote \emph{$j$-th subtracelets} of $T$, so that $T\equiv t_n\vert t_{n-1}\vert\dotsc\vert t_1$ is a concatenation of its subtracelets, with
\begin{equation}
	t_j:=\inputtikz{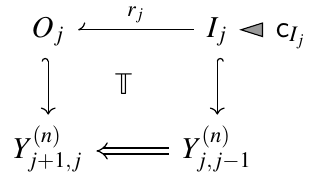}\,,\qquad 
	Y^{(n)}_{n+1,n}:=O_{n\cdots 1}\,,\; Y^{(n)}_{1,0}:=I_{n\cdots 1}\,.
\end{equation}
\end{definition}
Based upon the concurrency theorem for compositional rewriting theories (Theorem~\ref{thm:concur}), one may define an operation that provides the starting point of our tracelet analysis framework:
\begin{corollary}[Tracelet surgery]\label{cor:TAbase}
	Let $T\in \cT^{\bT}_n$ a $\bT$-type tracelet of length $n$, so that $T\equiv t_n|\dotsc|t_1$. Then for any consecutive subtracelets $t_j|t_{j-1}$ in $T$, one may uniquely (up to isomorphisms) construct a diagram $t_{(j\vert j-1)}$ and a tracelet $T_{(j\vert j-1)}$ of length $2$ as follows:
\begin{subequations}
\begin{align}
&\inputtikz{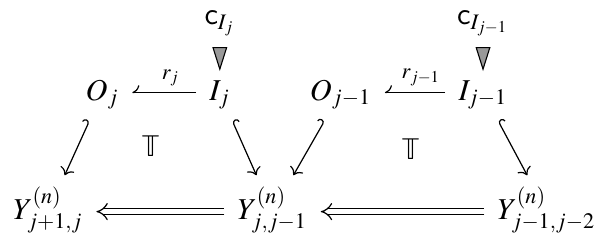}
	\;\rightsquigarrow\; 
\inputtikz{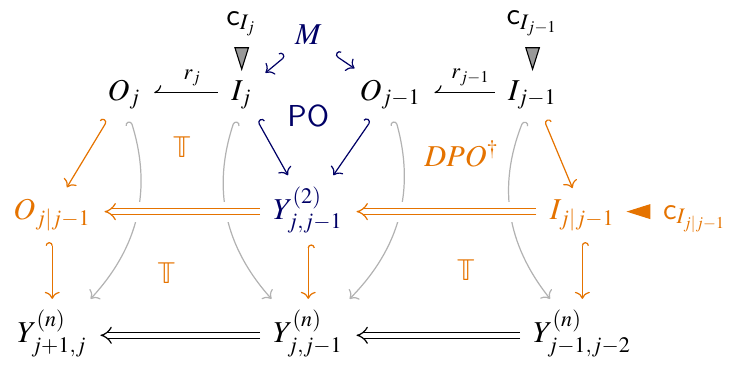}\label{eq:thmTA1}\\
&t_{(j\vert j-1)}:= 
\inputtikz{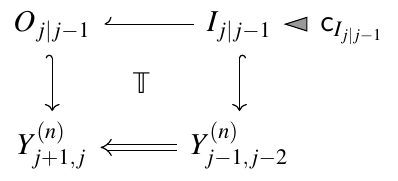}\,,\quad
T_{(j\vert j-1)}:=\TcompGT{T(r_j,\ac{c}_{I_j})}{{\color{h1color}\mu}}{T(r_{j-1},\ac{c}_{I_{j-1}})}{\bT}\label{eq:thmTA2}
\end{align}
\end{subequations}
Here, ${\color{h1color}\mu}=(I_j\hookleftarrow M\hookrightarrow O_{j-1})$ is the span of monomorphisms obtained by taking the pullback of the cospan $(I_j\hookrightarrow Y^{(n)}_{j,j-1}\hookleftarrow O_{j-1})$, and this ${\color{h1color}\mu}$ is always a $\bT$-admissible match. By associativity of the tracelet composition, this extends to consecutive sequences $t_j\vert\dotsc\vert t_{j-k}$ of subtracelets in $T$ inducing diagrams $t_{(j\vert\dotsc\vert j-k)}$ and tracelets of length $1$ $T_{(j\vert \dotsc\vert j-k)}$, where for $k=0$, $t_{(j)}=t_j$ and $T_{(j)}=T(r_j,\ac{c}_{I_j})$.
\end{corollary}
\begin{proof}
The proof follows by invoking Theorem~\ref{thm:tlChar} in order to convert the derivation trace encoded in a given the subdiagram $t_j\vert\dotsc\vert t_{j-k}$ into the application of a tracelet $T_{(j\vert \dotsc\vert j-k)}$ of length $k$ onto the initial object $Y^{(n)}_{j-k,j-k-1}$.
\end{proof}

It is via these tracelet surgery operations that we may lift the theory of static analysis of derivation traces to our alternative setting of tracelets. We define in the following two notions of equivalence that have analogues also in the traditional theory of rewriting systems~\cite{Kreowski1987,Corradini1994,ehrig1999handbook,ehrig2006fund}.

\begin{definition}[Tracelet abstraction equivalence]\label{def:tAeq}
	Two tracelets $T,T'\in \cT^{\bT}_n$ of the same length $n\geq 1$ are defined to be \emph{abstraction equivalent}, denoted $T\equivAbs T'$, if there exist suitable isomorphisms on the objects in $T$ in order to transform $T$ into $T'$ (with transformations on morphisms induced by object isomorphisms).
\end{definition}
Due to the intrinsic invariance of all category-theoretical constructions pertaining to rewriting rules as well as tracelets up to universal isomorphisms, it is clear that abstraction equivalence is a very natural\footnote{While we will typically consider tracelets by default only up to abstraction equivalence, the definitions provided thus far may nevertheless still be interpreted as concrete operations if suitable ``standard representatives'' are chosen for pushouts, pullbacks etc. --- for instance, an extensive discussion of such an interplay of concrete representatives vs.\ universal structures for the special case of graph rewriting systems may be found in~\cite{baldan2000modelling}.}, or even essential type of equivalence.

%%%%%%%%%%%%%%%%%%

\begin{figure}\gdef\gsScale{0.8}
  \begin{subfigure}{0.5\textwidth}
    \centering
    $\begin{array}{c}
      \includegraphics[scale=0.5]{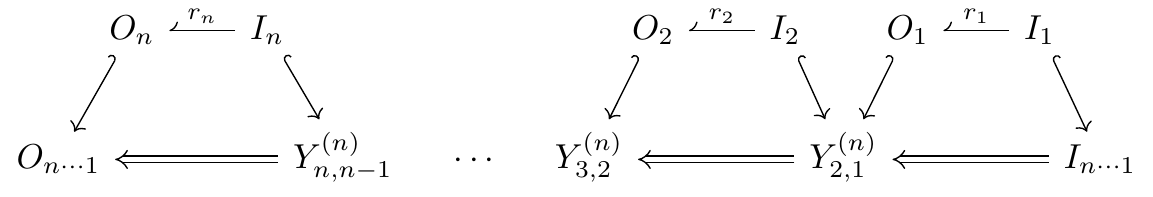}\\
      \widehat{=}\\
      \includegraphics[scale=0.5]{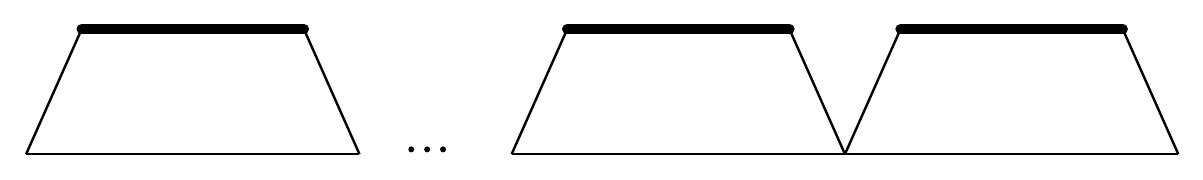}
    \end{array}$
    \caption{\label{fig:traceletoverviewA}Tracelets as (minimal) derivation traces.}
  \end{subfigure}\hfill
  \begin{subfigure}{0.5\textwidth}
    \centering
    \includegraphics[scale=0.3]{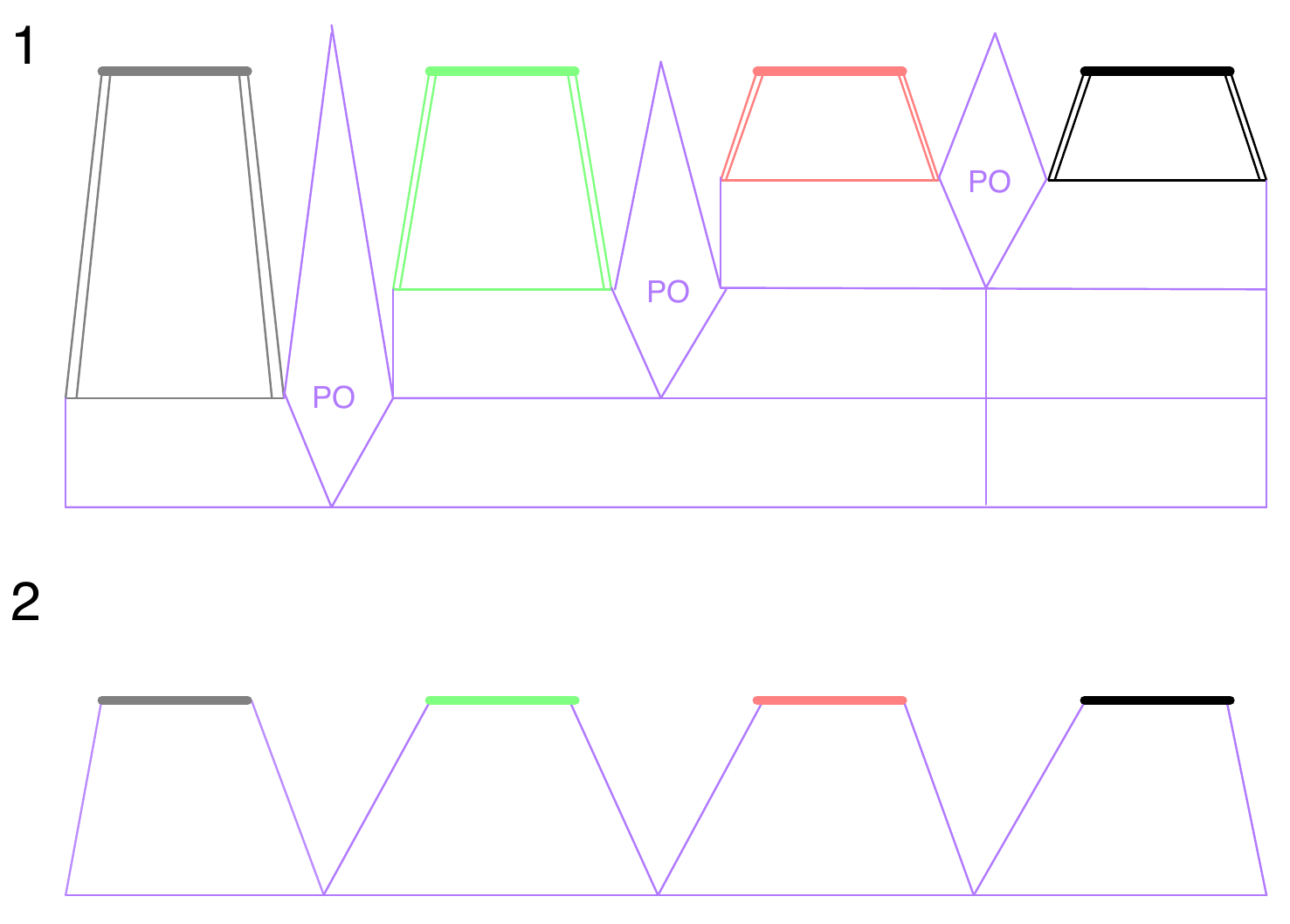}
    \caption{\label{fig:traceletoverviewC}Tracelet generation (Definition~\ref{def:genT}).}
  \end{subfigure}

  \begin{subfigure}{0.5\textwidth}
    \centering
    \includegraphics[scale=0.3]{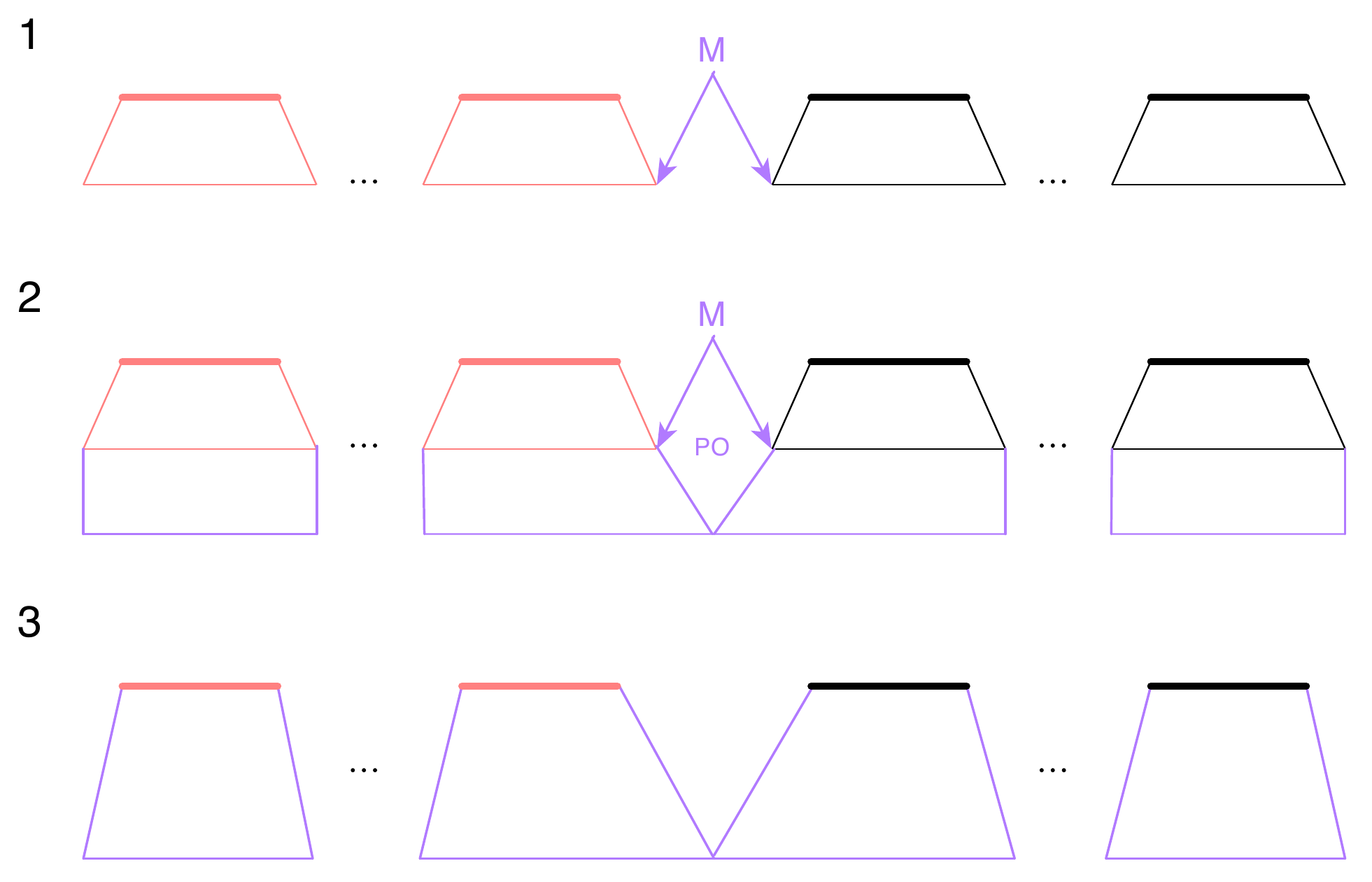}
    \caption{\label{fig:traceletoverviewB}Tracelet composition (Definition~\ref{def:tlComp}).}
  \end{subfigure}\hfill
  \begin{subfigure}{0.5\textwidth}
    \centering
    \includegraphics[scale=0.3]{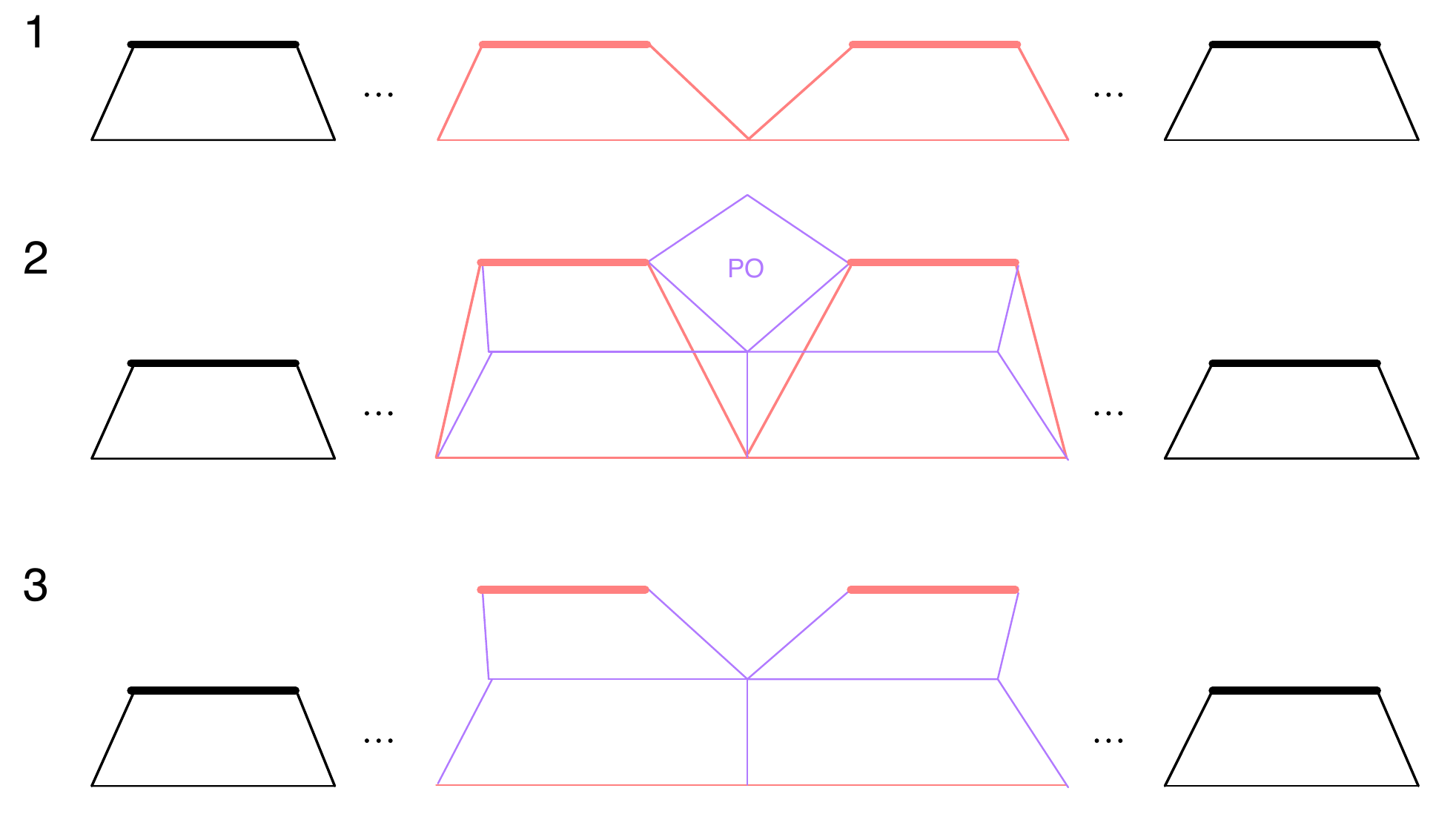}
    \caption{\label{fig:traceletoverviewD}Tracelet analysis (Section~\ref{sec:ta}).}
  \end{subfigure}
  \caption{Schematic overview of the tracelet and tracelet analysis framework.}\label{fig:traceletoverview}
\end{figure}

%%%%%%%%%%%%%%%

\begin{definition}[Tracelet shift equivalence]\label{def:tSeq}
	Let $T,T'\in\cT^{\bT}_n$ be two tracelets of the same length $n\geq 1$. If there exist subtracelets $t_j\vert \dotsc\vert t_{j-k}$ and $t'_j\vert \dotsc\vert t'_{j-k}$ such that
\begin{romanenumerate}
		\item the subtracelets have the same rule content (up to isomorphisms), i.e.\ there exists a permutation $\sigma\in S_{k+1}$ such that $[[T_{(p)}]]\cong [[T'_{(\sigma(p))}]]$ for all $j-k\leq p\leq j$, and
		\item the diagrams $t_1\vert\dotsc\vert t_{(j\vert\dotsc \vert j-k)}\vert \dotsc\vert t_n$ and $t_1'\vert\dotsc\vert t'_{(j\vert\dotsc \vert j-k)}\vert \dotsc\vert t_n'$ are isomorphic,
\end{romanenumerate}
\noindent then $T$ and $T'$ are defined to be \emph{shift equivalent}, denoted $T\equivShift T'$. Extending $\equivShift$ by transitivity then yields an equivalence relation on $\cT^{\bT}_n$ for every $n\geq 1$.
\end{definition}
Referring to Appendix~\ref{app:tlShiftEqExample} for the precise details, one may for example verify that the tracelet $t_3\vert t_2\vert t_1$ of length $3$ depicted in Fig.~\ref{fig:traceletAssociativityB} is shift equivalent to a tracelet $t_2'\vert t_3'\vert t_1$,  with the order of the applications of the rules $r_3$ and $r_2$ (contained in the yellow box in Fig.~\ref{fig:traceletAssociativityB}) reversed.  Notably, while our definition of tracelet abstraction equivalence follows precisely the same methodology as its analogous notion in rewriting theory, our definition of tracelet shift equivalence is strictly more general than the notion of shift equivalence in rewriting theories according to the standard literature~\cite{Kreowski1987,ehrig1999handbook,ehrig2006fund}. More precisely, the latter concept is based upon so-called \emph{sequential independence} for derivation sequences~~\cite{ehrig1999handbook,ehrig2014mathcal,Corradini2006,BALDAN2014}, which would induce a notion of tracelet shift equivalence strictly less permissive than our requirements described in Definition~\ref{def:tSeq}. As this difference is of crucial importance to the design of static analysis algorithms, we provide the precise technical relationship in Theorem~\ref{thm:CCC} of Appendix~\ref{app:CSI} for clarification.

\section{Application: a prototype for a Feature-driven Explanatory Tracelet Analysis (FETA) algorithm}\label{sec:FETA}

A major motivation behind the development of the tracelet analysis framework has been the desire to improve upon (and, to an extent, also formalize) existing static analysis techniques for rewriting systems in the application areas of bio- and organic chemical reaction systems (see also Section~\ref{sec:intro}). An application of our framework to the static generation of so-called \emph{pathways} appears to be particularly promising:
\begin{definition}[Pathways (sketch)]\label{def:pathways}
	Let $\cR=\{ R_j\in \LinAc{\bfC}\}_{j\in J}$ be a (finite) set of rules with conditions over $\bfC$, which model the \emph{transitions} of a rewriting system. We designate a rule $E\in \LinAc{\bfC}$ as modeling a \emph{``target event''}, i.e.\ $E$ must be the last rule applied in the derivation traces we will study. Let moreover $\equiv_C$ be an equivalence relation on derivation traces such as abstraction or shift equivalence, or combinations thereof. Then the task of \emph{pathway generation} or \emph{explanatory synthesis} for the type-$\bT$ rewriting system based upon the set of rules $\cR$ is defined as follows: synthesize the \emph{maximally compressed} derivation traces ending in an application of $E$ such that ``$E$ cannot occur at an earlier position in a given trace''. Here, compression refers to retaining only the smallest traces in a given ($\equiv_C$)-equivalence class, while the last part of the statement needs to be made precise in a specific application (as it depends on the chosen framework).
\end{definition}

A standard approach to this type of task consists in generating a large number of random generic derivation traces first, followed by static analysis type operations performed on these traces in order to extract pathways (see e.g.\ the recent review~\cite{Boutillier2018aa}). This type of approach typically suffers from two disadvantages: (i) depending on the complexity of the rule set $\cR$ and of the target event $E$, it may be difficult to find suitable choices of initial objects $X_0$ as an input to the simulation algorithms, and (ii) the extraction of compressed pathways from typically quite extensive datasets of simulator outputs may be computationally rather intense. We thus propose an alternative pathway generation approach based upon tracelets, which avoids the first problem by design (since tracelets are composable with themselves and yield the minimal derivation traces for entire classes of derivations according to Theorem~\ref{thm:tlChar}).

\begin{definition}[Algorithm~\ref{algo:FETA}: FETA]
With input data as described in Algorithm~\ref{algo:FETA}, let $\equiv_C$ be the equivalence relation obtained by conjunction of the tracelet abstraction and tracelet shift equivalences $\equivAbs$ and $\equivShift$, respectively. Then for a tracelet $T\in \cT^{\bT}_{n+1}$ of the structure $T=t_E\vert t_n\vert\dotsc \vert t_1$ (for some finite value $n\geq 1$, and with $t_E$ containing the rule $E$, $[[T_{(E)}]]\cong E$), we let $E\prec_C T$ denote the following property: there exist no tracelets $T'\in \cT^{\bT}_{n+1}$
\begin{equation}
		t_E \vert t_n\vert \dotsc \vert t_1 \equiv_C
		t_{n+1}' \vert t_n'\vert \dotsc \vert t_1' 
    \quad \text{with }\;[[T'_{(k)}]]\cong E\;\text{for  an index $k<n+1$}\,.
\end{equation}
We refer to the set of such tracelets modulo $\equiv_C$ as the set of \emph{strongly compressed pathways}.
\end{definition}

\begin{algorithm}
    \SetAlgoLined
    \KwData{$N_{max}\geq 2$ $\gets$ maximal length of tracelets to be generated\\
    $T_E:=T(E)$ $\gets$ tracelet of length $1$ associated to the rule $E$\\
    $\mathsf{T}_1:=\{T(R_j)\mid j\in J\}$ $\gets$ set of tracelets of length $1$ associated to the transitions}
    \KwResult{sets $\mathsf{P}_i$ ($i=2,\dotsc,N_{max}$) of strongly compressed pathways}
    \Begin{
    $P_1:=\{T_E\}$ $\gets$ the only pathway of length $1$;\\
    \For{$2< n\leq N_{max}$}{
	   $\mathsf{pre}_n:=
    \left.\left\{\TcompGT{P}{\mu}{T}{\bT}\right\vert 
    	P\in \mathsf{P}_{n-1}, T\in \mathsf{T}_1\,,\; 
    	\mu\in \tMatchGT{P}{T}{\bT}\right \}$\;
    $\mathsf{P}_n:=\{T'\in \mathsf{pre}_n\vert E\prec_C T'\}\diagup_{\equiv_C}$\;
    }
    }
\caption{Feature-driven Explanatory Tracelet Analysis (FETA)\label{algo:FETA}}
\end{algorithm}

Since length limitations preclude presenting an application example of realistic complexity in one of the chemical reactions system frameworks, we will present here only a first proof of concept for an application of the FETA algorithm, which nevertheless illustrates in which sense the above algorithm synthesizes ``explanations''.

\begin{example}\label{ex:FETA}
	Let $\bfC=\mathbf{FinGraph}$ be the category of finite directed multigraphs. For compactness of graphical illustrations and to enhance intuitions, we will present linear rules $r=(O\hookleftarrow K\hookrightarrow I)\in \Lin{\mathbf{FinGraph}}$ in a diagrammatic form, where graphs $O$ and $I$ are depicted to the left and to the right, respectively, and where dotted lines connecting elements of $I$ with elements of $O$ indicate the structure of the partial map encoded in the span $r$. Let thus $\cR=\{r\}$ be a one-element transition set (for a rule $r\in \Lin{\mathbf{FinGraph}}$ without conditions), and let $e_1,e_2\in \Lin{\mathbf{FinGraph}}$ be two rules modeling alternative target events: 
	\begin{equation}
		r=\vcenter{\hbox{\rotatebox{90}{
		\inputtikzNoVcenter{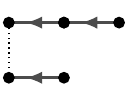}
		}}}\,,\quad
		e_1=\vcenter{\hbox{\rotatebox{90}{
		\inputtikzNoVcenter{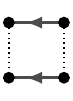}}}}\,,\quad
		e_2=\vcenter{\hbox{\rotatebox{90}{
		\inputtikzNoVcenter{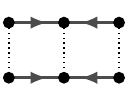}}}}\,.
	\end{equation}
If we consider DPO-type rewriting, the FETA algorithm produces the following strongly compressed pathways for $n\geq 2$ (with {\color{h3color}light blue} arrows indicating the relative overlap structure within the tracelets):
\begin{equation}
	\mathsf{P}_n=\{S_n\}\,,\qquad 
	S_n=t_E\vert \!\!\underbrace{t_r\vert \dotsc\vert t_r}_{\text{$(n-1)$ times}}=
		\vcenter{\hbox{\rotatebox{90}{
		\inputtikzNoVcenter{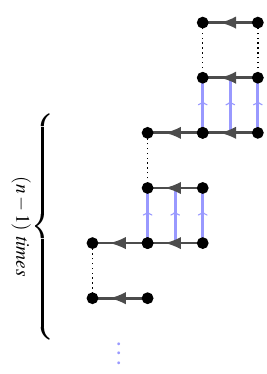}}}}\,,
	\end{equation}
	while for the target event $e_2$ the algorithm detects no pathways $\mathsf{P}'_n$ for $n\geq 2$. This result may indeed be interpreted as expressing a high-level causal structure or explanation about this simple rewriting system. As for $e_1$, the pathways $\mathsf{P}_n$ are seen to effectively encode those possibilities of sequential rule compositions that ensure that the edge eventually matched by $e_1$ had not already been present in any of the first $n-2$ steps of rule applications. This leaves only the pathways of type $S_n$ as options, since for any other match of the tracelet $T_E$ within a candidate tracelet $T$ of length $n$, one finds a violation of the condition $E\prec_C T$. On the other hand, the fact that there are no pathways of length $n\geq 2$ for the target event encoded by $e_2$ signifies that the rule $r$ acting on some initial graph $X_0$ can in fact not generate any occurrences of the shape of the input of $e_2$ (two edges with a shared vertex pointing towards each other) that had not already been present in $X_0$. Note that we have obtained this result \emph{statically}, and without ever evaluating any concrete direct derivation on initial graphs $X_0$.
\end{example}

\section{Conclusion and Outlook}\label{sec:conclusion}

Many of the standard constructions in the concurrency theory and the theory of static analysis of rewriting systems over adhesive categories are ultimately based upon one of the central theorems of rewriting theory, which is known fittingly as the \emph{concurrency theorem}~\cite{DBLPconf/gg/1997handbook,ehrig2006fund,ehrig2014mathcal}. The essential property provided by this theorem is a form of compatibility between (i) sequential applications of rewriting rules starting at some initial object $X_0$, and (ii) a one-step application of a \emph{composition} of the rewriting rules involved, and with both descriptions in a (constructive) bijective correspondence. As outlined in Section~\ref{sec:intro}, it is then precisely this correspondence which allows to develop various abstractions and analysis techniques for derivation traces of a given rewriting system~\cite{Baldan1998aa,Baldan2006aa,baldan2000modelling,Baldan2007aa,danos2012graphs,BALDAN2014}. However, as has been only very recently discovered~\cite{bdg2016,bp2018,nbSqPO2019,bk2019a}, both Double-Pushout (DPO) and Sesqui-Pushout (SqPO) rewriting theories over suitable adhesive categories carry an additional important structure, namely on the operation of \emph{composing rules} itself: in a certain sense, rule compositions are \emph{associative} (with a concrete example provided in Figure~\ref{fig:traceletAssociativity}). 

In this paper, we demonstrate that combining the concurrency with the associativity theorems, one is naturally led to the concept of \emph{tracelets} (Section~\ref{sec:tSetup}), which may be intuitively understood as a form of \emph{minimal derivation traces} that generate all derivations that are based upon the same sequential rule compositions (Theorem~\ref{thm:tlChar}). Owing to the associativity theorem for rule compositions, tracelets are on the one hand by definition instances of derivation traces themselves and thus admit all aforementioned standard static analysis techniques, but importantly in addition afford certain universal properties: an associative notion of composition directly on tracelets, certain types of ``surgery'' operations, and finally various forms of equivalence relations that may be employed to develop compressions and other abstractions of tracelets (Section~\ref{sec:ta}). 

In view of practical applications, we have proposed a first prototypical tracelet-based static analysis algorithm, the so-called \emph{Feature-driven Explanatory Tracelet Analysis (FETA)} algorithm (Section~\ref{sec:FETA}). As illustrated in Example~\ref{ex:FETA}, this algorithm permits to extract high-level causal information on the ``pathways'' or minimal derivation traces that can lead to the ultimate application of the rule that models a target event. We believe that our methodology may provide a significant contribution to the static analysis toolset in the fields of chemical graph transformation systems~\cite{Benk2003,banzhafetalDR20154968,Andersen2016,andersen2018rule} and of rule-based modeling approaches to biochemical reaction systems such as the \textsc{Kappa}~\cite{Danos2003ab,Danosaa,Danosab,Boutillier2018aa} and the \textsc{BioNetGen}~\cite{Harris2016aa,Blinov2004aa} frameworks. 

%%
%% Bibliography
%%
 \bibliographystyle{eptcs}
 \bibliography{refs-NB-ACT-2019}

%%%%%
\newpage
\appendix
%%%%%

\section{Background material: compositional rewriting theories}
\label{app:CR}

For the readers' convenience, we collect in this section a number of technical results and details on rewriting theories, most of which is either standard material from the rewriting theory, or quoted from our recent series of works~\cite{bp2018,nbSqPO2019,bk2019a}.

\subsection{Properties of adhesive categories}\label{app:ACprops}

\begin{theorem}[\cite{bk2019a}]\label{thm:propsC}
	Let $\bfC$ be a category satisfying Assumption~\ref{as:main}. Then the following properties hold:
	\begin{romanenumerate}
		\item $\bfC$ has \emph{effective unions} (compare~\cite{lack2005adhesive}): given a commutative diagram as in the left of Fig.~\ref{fig:cdA}, if the $(b',c')$ is the pullback of the cospan of monomorphisms $(b,c)$ (which by stability of monos under pullback entails that $b',c'\in \mono{\bfC}$), and if $(e,f)$ is the pushout of the span $(b',c')$ (with $e,f\in \mono{\bfC}$ by stability of monos under pushout), then the morphism $d$ which exists by the universal property of pushouts is also a monomorphism. 
		\item Properties of \emph{final pullback complements (FPCs)}\footnote{Recall e.g.\ from~\cite{Corradini2006} that for a pair of composable morphisms $(c,a)$ such as in the middle part of Fig.~\ref{fig:cdA}, a pair of composable morphisms $(d,b)$ is an FPC of $(c,a)$ if $(a,b)$ is the pullback of $(c,d)$, and if the following universal property holds: given a cospan $(c,z)$ such that $(x,y)$ is the pullback of $(c,z)$ and such that there exists a morphism $w$ satisfying $z=a\circ w$, then there exists a unique (up to isomorphism) $w^{*}$ such that $z=d\circ w^{*}$.} in $\bfC$ (cf.\ the middle diagram in Fig.~\ref{fig:cdA}): for every pair of composable monomorphisms $(c,a)$, there exists an FPC $(d,b)$, and moreover $b,d\in \mono{\bfC}$. 
		\item \emph{Characterization of epimorphisms via pushouts}: given a diagram such as on the right of Fig.~\ref{fig:cdA}, where all morphisms except $e$ are monomorphisms, where the square $(1)$ is a pushout, $e\circ d_i=e_i$ for $i=1,2$, and where $(b_1,b_2)$ is the pullback of $(e_1,e_2)$. Then the morphism $e$ is an epimorphism if and only if the exterior square is a pushout.\label{thm:charEpi}
	\end{romanenumerate}
\end{theorem}

\begin{figure}[h]
\centering
\inputtikzNoVcenter{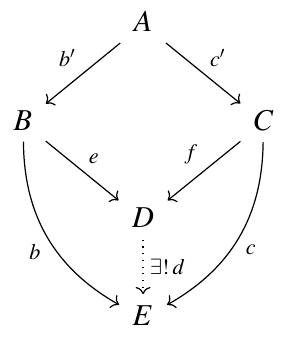}
$\qquad\qquad$
\inputtikzNoVcenter{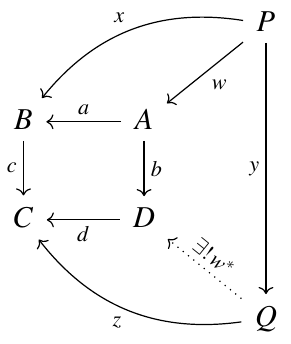}
$\qquad\qquad$
\inputtikzNoVcenter{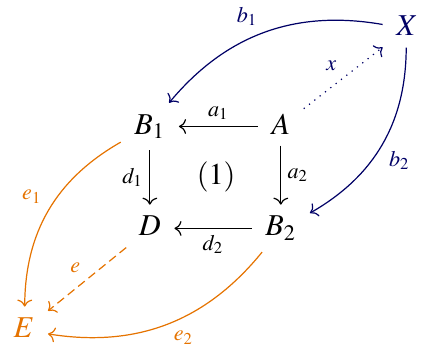}
\caption{\label{fig:cdA}Effective unions (left), final pullback complements (FPCs) and their universal property (middle), and the epimorphism-pushout correspondence (right)~\cite{bk2019a}.}
\end{figure}

\subsection{Conditions, shift and transport constructions}\label{sec:conds}

\begin{definition}[Conditions]
	Let $\bfC$ be a category satisfying Assumption~\ref{as:main}. Then a \emph{condition over an object $X\in \obj{\bfC}$}, denoted $\ac{c}_{X}$, is inductively defined as follows:
	\begin{itemize}
		\item $\ac{c}_X=\ac{true}$ is a condition over $X$.
		\item For every $(a:X\rightarrow A)\in \mono{\bfC}$ and for every condition $\ac{c}_A$ over $A\in \obj{\bfC}$, $\exists(a:X\rightarrow A,\ac{c}_A)$ is a condition over $X$.
		\item If $\ac{c}_X$ is a condition over $X$, so is its negation $\neg \ac{c}_X$.
		\item If $\ac{c}_X^{(i)}$ are conditions over $X$ (for indices $i\in I$), then $\land_{i\in I}\ac{c}_{X}^{(i)}$ is a condition over $X$.
	\end{itemize}
\end{definition}

A concrete interpretation of conditions is provided by the accompanying definition of \emph{satisfaction} of conditions.
\begin{definition}[Satisfaction]
Let $X\in \obj{\bfC}$ be an object and $\ac{c}_X$ a condition over $X$. Then the \emph{satisfaction of $\ac{c}_X$ by a monomorphism $(m:X\rightarrow Y)\in \mono{\bfC}$}, denoted $m\vDash \ac{c}_X$, is inductively defined as follows:
\begin{itemize}
		\item Every morphism satisfies $\ac{c}_X=\ac{true}$.
		\item For every $(a:X\rightarrow A)\in \mono{\bfC}$ and for every condition $\ac{c}_A$ over $A\in \obj{\bfC}$, the morphism $m:X\rightarrow Y$ satisfies $\exists(a:X\rightarrow A,\ac{c}_A)$ if there exists $(q:A\rightarrow Y)\in \mono{\bfC}$ such that $m=q\circ a$ and $q\vDash \ac{c}_A$.
		\item $m$ satisfies $\neg \ac{c}_X$ if it does not satisfy $\ac{c}_X$.
		\item If $\ac{c}_X^{(i)}$ are conditions over $X$ (with $i\in I$), $m$ satisfies $\land_{i\in I}\ac{c}_{X}^{(i)}$ if $m\vDash\ac{c}_X^{(i)}$ for all $i\in I$.
	\end{itemize}
\end{definition}

The notion of satisfaction of conditions permits to reason on equivalences of conditions:
\begin{definition}[Equivalence]
	Let $X\in \obj{\bfC}$ be an object, and let $\ac{c}_X^{(1)}$ and $\ac{c}_X^{(2)}$ be two conditions over $X$. Then the two conditions are \emph{equivalent}, denoted $\ac{c}_X^{(1)}\equiv \ac{c}_X^{(2)}$, iff
	\begin{equation}
		\forall (m:X\rightarrow Y)\in \mono{\bfC}:\; m\vDash \ac{c}_X^{(1)} \Leftrightarrow m\vDash \ac{c}_X^{(2)}\,.
	\end{equation}
\end{definition}

Besides the evident equivalences that arise from isomorphisms of rules and objects, some important classes of equivalences are implemented by the following two constructions quoted from~\cite{bk2019a}, which are essential in our compositional rewriting framework (cf.\ Section~\ref{sec:DPOandSqPOrewriting}).

\begin{definition}[Shift and Transport]
	Let $\bfC$ be a category satisfying Assumption~\ref{as:main}. Then for every condition $\ac{c}_A$ over $A\in \obj{\bfC}$ and for every $(a_1:A\rightarrow B_1)$, the \emph{shift of the condition $\ac{c}_A$ over $a_1$}, denoted $\Shift(a_1,\ac{c}_A)$, is defined inductively as follows:
	\begin{itemize}
		\item $\Shift(a_1,\ac{true}):=\ac{true}$ (over $B_1$).
		\item If $\ac{c}_A=\exists(a_2:A\rightarrow B_2,\ac{c}_{B_2})$ for some $(a_2:A\rightarrow B_2)\in \mono{\bfC}$, then
		\begin{equation}
			\Shift(a_1:A\rightarrow B_1,\exists(a_2,\ac{c}_{B_2})):=\bigvee_{(b_1,b_2)\in \cX}\exists(e_1:B_1\rightarrow E,\Shift(e_2:B_2\rightarrow E,\ac{c}_{B_2}))\,,
		\end{equation}
		where the set $\cX$ is the set of all isomorphism classes\footnote{Note that our improvement over the original variant of this construction as presented in~\cite{habel2009correctness,ehrig2014mathcal} consists in the precise characterization of the contributions to $\Shift(a_1,\exists(a_2,\ac{c}_{B_2}))$ via constructing pushouts of the possible spans $(b_1,b_2)$ (rather than via the original indirect characterization in terms of listing the possible epimorphisms $(e:D\rightarrow E)\in \epi{\bfC}$), a result which relies on Theorem~\ref{thm:propsC}\ref{thm:charEpi}, and which is of central importance to proving the compositionality and associativity of rules with conditions.} of spans $(b_1,b_2)$ as in the right part of Figure~\ref{fig:cdA}, where $(e_1,e_2)$ are constructed as the pushout of $(b_1,b_2)$.
		\item $\Shift(a_1,\neg\ac{c}_A):=\neg\Shift(a_1,\ac{c}_A)$  and $\Shift(a_1,\lor_{i\in I}\ac{c}_A^{(i)}):=\lor_{i\in I}\Shift(a_1,\ac{c}^{(i)}_A)$.
	\end{itemize}

	We also define the operations of \emph{transporting a condition $\ac{c}_O$ over the output object $O$ of a linear rule $r=(O\leftarrow K\rightarrow I)\in \Lin{\bfC}$ to the input object $I$ of $r$}.  The construction is denoted $\Trans(r,\ac{c}_O)$ and is defined inductively as follows:
	\begin{itemize}
		\item $\Trans(r,\ac{true}):=\ac{true}$ (over $I$).
		\item If $\ac{c}_O=\exists(b:O\rightarrow B,\ac{c}_B)$ for some $b\in \mono{\bfC}$, then if $b\not\in \MatchGT{DPO^{\dag}}{r}{B}$, we define $\Trans(r,\exists(b:O\rightarrow B,\ac{c}_B)):=\ac{false}$ (as a condition over $I$). Otherwise, i.e.\ if $b\in \MatchGT{DPO^{\dag}}{r}{B}$, we let (referring to Definition~\ref{def:rew} for the definition of $DPO^{\dag}$)
\begin{equation}
	\Trans(r,\exists(b:O\rightarrow B,\ac{c}_B))
		:=\exists({\color{h2color}b^{*}:I\rightarrow B'},
		\Trans(B\Leftarrow {\color{h2color}B'},\ac{c}_B)\,,\quad
	\text{where }\quad
	\inputtikz{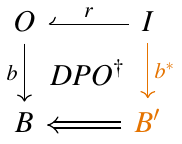}\,.
\end{equation}
	\item $\Trans(r,\neg \ac{c}_O):=\neg \Trans(r,\ac{c}_O)$ and $\Trans(r,\lor_{i\in I}\ac{c}_O^{(i)}):=\lor_{i\in I}\Trans(r,\ac{c}_O^{(i)})$.
	\end{itemize}
\end{definition}

\begin{theorem}[Properties of shift and transport constructions~\cite{bk2019a}; compare~\cite{GOLAS2014}]
Let $\bfC$ be a category satisfying Assumption~\ref{as:main}.
\begin{romanenumerate}
\item \emph{$\Shift$ and satisfaction:} for $X\in \obj{\bfC}$, $\ac{c}_X$ a condition over $X$ and $(m:X\rightarrow Y)\in \mono{\bfC}$, then for monomorphisms $(q:Y\rightarrow Z)\in \mono{\bfC}$ it holds that
\begin{equation}
	q\vDash\Shift(m:X\rightarrow Y,\ac{c}_X)\;\Leftrightarrow \;
	q\circ m\vDash \ac{c}_X\,.
\end{equation}
\item \emph{Unit for $\Shift$:} for every object $X\in \obj{\bfC}$ and for every $(X\xrightarrow{\cong} X')\in\iso{\bfC}$, 
\begin{equation}
	\Shift(X\xrightarrow{\cong} X',\ac{c}_X)\equiv \ac{c}_X\,.
\end{equation}
\item \emph{Compositionality of $\Shift$:} given composable monomorphisms $(f:X\rightarrow Y),(g:Y\rightarrow Z)\in \mono{\bfC}$ and a condition $\ac{c}_X$ over $X$, 
\begin{equation}
	\Shift(g:Y\rightarrow Z,\Shift(f:X\rightarrow Y,\ac{c}_X))\equiv \Shift(g\circ f:X\rightarrow Z,\ac{c}_X)\,.
\end{equation}
\item \emph{Satisfiability for $\Trans$:} for $r=(O\leftarrow K\rightarrow I)\in \Lin{\bfC}$, $\ac{c}_O$ a condition over $O$ and $X\in \obj{\bfC}$, denoting for an admissible match $m\in\MatchGT{\bT}{r}{X}$ (for $\bT\in\{DPO,SqPO\}$) the comatch for $m$ under a $\bT$-type rule application by $m^{*}$ (cf. \eqref{eq:DD}),
\begin{equation}
	\forall m\in \MatchGT{\bT}{r}{X}:\quad m\vDash\Trans(r,\ac{c}_O) \Leftrightarrow m^{*}\vDash \ac{c}_O\,.
\end{equation}
We write $\Trans(r,\ac{c}_O)\,\dot{\equiv}_{\bT}\,\ac{c}_O$ to indicate this \emph{equivalence up to $\bT$-type admissibility}.
\item \emph{Units for $\Trans$:} for each span of isomorphisms $(Z\xleftarrow{\cong}Y\xrightarrow{\cong}X)\in \Lin{\bfC}$ and for each condition $\ac{c}_Z$ over $Z$, 
\begin{equation}
	\Trans((Z\xleftarrow{\cong}Y\xrightarrow{\cong}X),\ac{c}_X)\equiv \ac{c}_Z\,.
\end{equation}
\item \emph{Compositionality of $\Trans$:} for composable spans $s=(E\leftarrow D\rightarrow C),r=(C\leftarrow B\rightarrow A)\in \Lin{\bfC}$ and $\ac{c}_E$ a condition over $E$, 
\begin{equation}
	\Trans(r,\Trans(s,\ac{c}_E)\equiv \Trans(s\circ r,\ac{c}_E)\,.
\end{equation}
\item \emph{Compatibility of $\Shift$ with $\Trans$:} for $r=(O\leftarrow K\rightarrow I)\in \Lin{\bfC}$, $(m:I\rightarrow X)\in \mono{\bfC}$, $\ac{c}_O$ a condition over $O$ and $\bT\in\{DPO,SqPO\}$,
\begin{equation}
m\in \MatchGT{\bT}{r}{X}\; \Rightarrow 
\Shift(m,\Trans(r,\ac{c}_O))\;\dot{\equiv}_{\bT}\; \Trans((r_m(X)\xLeftarrow[r,m]{\bT} X),\Shift(m^{*},\ac{c}_O))\,.
\end{equation}
\end{romanenumerate} 
\end{theorem}

\subsection{Compositional Double- and Sesqui-Pushout rewriting}
\label{sec:DPOandSqPOrewriting}

We present here our recently introduced refinement of the DPO-type such framework, as well as the first of its kind such framework for compositional SqPO-type rewriting theories~\cite{bk2019a}. For the DPO-type case, our developments generalized earlier work by Habel and %
Pennemann~\cite{Pennemannaa,habel2009correctness} and by Ehrig and %
collaborators~\cite{ehrig2014mathcal,GOLAS2014}, while the SqPO-constructions were new. The core definitions of rewriting are the following two sets of definitions, which provide the semantics of rule applications to objects and compositions of linear rules:

\begin{definition}[Rules and rule applications]\label{def:rew}
Let $\bfC$ be a category satisfying Assumption~\ref{as:main}. Denote by $\Lin{\bfC}$ the set of isomorphism classes of spans of monomorphisms,
\begin{equation}
	\Lin{\bfC}:=\left.\left\{ O\xleftarrow{o}K\xrightarrow{i}I \right\vert o,i\in \mono{\bfC}\right\}\diagup_{\cong}\,,
\end{equation}
referred to henceforth as the \emph{set of linear rules (on $\bfC$)}. Here, two rules $r_i=(O_j\leftarrow K_j\rightarrow I_j)\in \Lin{\bfC}$ (for $j=1,2$) are representatives of the same isomorphism class if there exist isomorphisms $\omega:O_1\rightarrow O_2$, $\kappa:K_1\rightarrow K_2$ and $\iota:I_1\rightarrow I_2$ that make the evident diagram commute\footnote{We will henceforth speak about isomorphism classes of rules and objects and their concrete representatives interchangeably, since all constructions presented afford a clear notion of invariance under isomorphisms. This feature of the constructions also motivates the notion of \emph{abstraction equivalence} standard in the rewriting theory (cf.\ Section~\ref{sec:ta}).}. Let $r=(O\leftarrow K\rightarrow I)\in \Lin{\bfC}$ be a linear rule,  $X\in \obj{\bfC}$ an object, and $(m:I\rightarrow X)\in \mono{\bfC}$ be a monomorphism. A \emph{rule application} of the rule $r$ to the object $X$ along an admissible match $m$ is defined via the following type of commutative diagram (referred to as a \emph{direct derivation} in the literature):
\begin{equation}\label{eq:DD}
\inputtikz{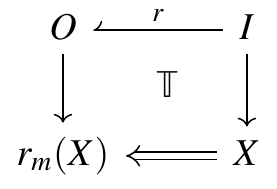}\quad :=\quad 
\inputtikz{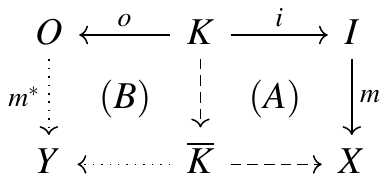}\,.
\end{equation}
Here, the precise construction depends on the \emph{type} $\bT$ of rewriting:
\begin{romanenumerate}
	\item \emph{$\bT=DPO$ (\textbf{Double-Pushout (DPO) rewriting}):} given $(m:I\rightarrow X)\in\mono{\bfC}$, $m$ is defined to be an \emph{admissible match} if the square marked $(A)$ in~\eqref{eq:DD} is constructable as a \emph{pushout complement (POC)}; in this case, the square marked $(B)$ is constructed as a \emph{pushout (PO)}.
	\item \emph{$\bT=DPO^{\dag}$ (\textbf{``opposite'' of Double-Pushout rewriting}):} for this auxiliary rewriting semantics, given $(m^{*}:O\rightarrow Y)\in\mono{\bfC}$, $m^{*}$ is defined to be an \emph{admissible match} if the square marked $(B)$ in~\eqref{eq:DD} is constructable as a \emph{pushout complement (POC)}; in this case, the square marked $(A)$ is constructed as a \emph{pushout (PO)}. Coincidentally, a $DPO^{\dag}$-admissible match $m^{*}$ uniquely induces a $DPO$-admissible match $m$ for $r$ applied to $X$, which is why $m^{*}$ is sometimes referred to as the \emph{comatch of $m$}.
	\item \emph{$\bT=SqPO$ (\textbf{Sesqui-Pushout (SqPO) rewriting}):} given $(m:I\rightarrow X)\in\mono{\bfC}$, the square $(A)$ is constructed as a \emph{final pullback complement (FPC)}, followed by constructing $(B)$ as a pushout.
\end{romanenumerate}%
\noindent We denote the \emph{set of $\bT$-admissible matches} by $\MatchGT{\bT}{r}{X}$. We will moreover adopt the traditional ``direct derivation'' notation $r_m(X) \xLeftarrow[r,m]{\bT} X$ for $\bT=DPO$ or $\bT=SqPO$ as a compact notation for the process of applying rule $r$ to $X$ along admissible match $m$.
\end{definition}

Crucially, linear rules in both types of rewriting semantics admit a \emph{composition operation}:

\begin{definition}[Rule composition]\label{def:Rcomp}
Let $\bfC$ be a category satisfying Assumption~\ref{as:main}. Let $r_1,r_2\in \Lin{\bfC}$ be two linear rules. We define the \emph{set of $\bT$-type admissible matches of $r_2$ into $r_1$} for $\bT\in \{DPO,SqPO\}$, denoted $\MatchGT{\bT}{r_2}{r_1}$, as
\begin{equation}
\begin{aligned}
	\MatchGT{\bT}{r_2}{r_1}&:=\big\{
		{\color{h1color}\mu_{21}}={\color{h1color}(I_2\leftarrow M_{21}\rightarrow O_2)}\big \vert n_1,n_2 \text{ in }
		\pO{{\color{h1color}\mu_{21}}}={\color{h1color}(I_2\xrightarrow{n_2}N_{21}\xleftarrow{n_1}O_1)}\\
		&\qquad\qquad \text{satisfy } n_2\in \MatchGT{\bT}{r_2}{{\color{h1color}N_{21}}}\;\land \;
		n_1\in \MatchGT{DPO^{\dag}}{r_1}{{\color{h1color}N_{21}}}
	\big\}\,.
\end{aligned}
\end{equation}
For a $\bT$-type admissible match ${\color{h1color}\mu_{21}}={\color{h1color}(I_2\leftarrow M_{21}\rightarrow O_2)}\in\MatchGT{\bT}{r_2}{r_1}$, construct the  diagram
\begin{equation}\label{eq:defRcomp}
\inputtikz{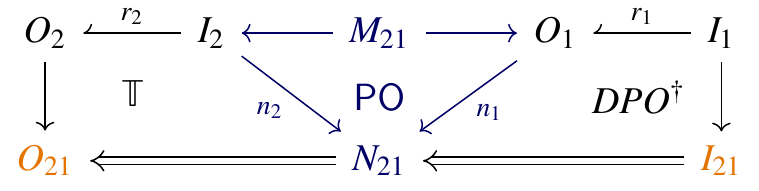}\,.
\end{equation}
From this diagram, one may compute (via pullback composition $\circ$ of the two composable spans in the bottom row) a span of monomorphisms ${\color{h2color}(O_{21}\Leftarrow I_{21})}\in\Lin{\bfC}$, which we define to be the \emph{$\bT$-type composition of $r_2$ with $r_1$ along $\mu_{21}$} (for $\bT\in\{DPO,SqPO\}$ as in~\eqref{eq:defRcomp}):
\begin{equation}
	\compGT{\bT}{r_2}{\mu_{21}}{r_1}:={\color{h2color}(O_{21}\Leftarrow I_{21})}=({\color{h2color}O_{21}}\Leftarrow {\color{h1color}N_{21}})\circ ({\color{h1color}N_{21}}\Leftarrow{\color{h2color}I_{21}})\,.
\end{equation}
\end{definition}

We refer the interested readers to~\cite{bp2018,nbSqPO2019} for the precise derivations of these notions of rule compositions, and note here that the definitions are justifiable a posteriori via the concurrency theorems as presented in Section~\ref{sec:concur}.

The compositional rewriting framework may be extended to the setting of rules with conditions\footnote{Referring to~\cite{bk2019a} for the precise details, based upon the transport construction we may from hereon without loss of generality consider rewriting rules with conditions over the \emph{input} objects only.}:
\begin{definition}[Rewriting with conditions~\cite{bk2019a}]\label{def:RwCs}
	Let $\bfC$ be a category satisfying Assumption~\ref{as:main}. We denote by $\LinAc{\bfC}$ the \emph{set of linear rules with conditions}, defined as
	\begin{equation}
		\LinAc{\bfC}:=\{R=(r,\ac{c}_I)\mid r=(O\leftarrow K\rightarrow I)\in \Lin{\bfC}\}\,.
	\end{equation}
	We extend the definitions of rule applications (Definition~\ref{def:rew}) and rule compositions (Definition~\ref{def:Rcomp}) to the setting of rules with conditions as follows: for $R=(r,\ac{c}_I)\in \LinAc{\bfC}$ and $X\in \obj{\bfC}$, define the \emph{sets of $\bT$-type admissible matches of $R$ into $X$} (for $\bT\in\{DPO,SqPO\}$), denoted $\MatchGT{\bT}{R}{X}$, as
	\begin{equation}
		\MatchGT{\bT}{R}{X}:=\{m\in \MatchGT{\bT}{r}{X}\mid m\vDash\ac{c}_I\}\,.
	\end{equation}
	Then the \emph{$\bT$-type rule application of $R$ along $m\in \MatchGT{\bT}{R}{X}$ to $X$} is defined as $r_m(X)\xLeftarrow{\bT}X$. 

	As for the rule compositions, we define for two rules with application conditions $R_j=(r_j,\ac{c}_{I_j})\in \LinAc{\bfC}$ ($j=1,2$) the \emph{sets of $\bT$-admissible matches of $R_2$ into $R_1$} as
	\begin{equation}
		\MatchGT{\bT}{R_2}{R_1}:=\{\mu_{21}\in \MatchGT{\bT}{r_2}{r_1}\mid \ac{c}_{I_{21}}\not\equiv \ac{false}\}\,,
	\end{equation}
	where the condition $\ac{c}_{I_{21}}$ for a given rule composite is defined as (compare~\eqref{eq:defRcomp} for the defining construction of the various morphisms)
	\begin{equation}
	\begin{aligned}
		\ac{c}_{I_{21}}&:= \Shift(I_1\rightarrow {\color{h2color}I_{21}},\ac{c}_{I_1}) \bigwedge
		\Trans({\color{h1color}N_{21}}\Leftarrow {\color{h2color}I_{21}},\Shift(I_2\rightarrow {\color{h1color}N_{21}},\ac{c}_{I_2}))\,.
	\end{aligned}
	\end{equation}
	Then we define for admissible matches $\mu_{21}$ the compositions as
	\begin{equation}
		\forall \mu_{21}\in \MatchGT{\bT}{R_2}{R_1}:\quad \compGT{\bT}{R_2}{\mu_{21}}{R_1}:=(\compGT{\bT}{r_2}{\mu_{21}}{r_1},\ac{c}_{I_{21}})\,.
	\end{equation}
\end{definition}

\subsection{Auxiliary properties of direct derivations}\label{app:aux}

\begin{lemma}\label{lem:sqCompAux}
Let $\bfC$ be a category satisfying Assumption~\ref{as:main}, and let $\bT\in\{DPO,SqPO,DPO^{\dag}\}$. 
\begin{romanenumerate}
  \item\label{lem:Auxa} For every rule $R=(r,\ac{c}_I)\in\LinAc{\bfC}$, the diagram below is a $\bT$-type direct derivation for arbitrary $\bT\in \{DPO,DPO^{\dag},SqPO\}$:
  \begin{equation}\label{eq:DDaux}
\left(
\inputtikz{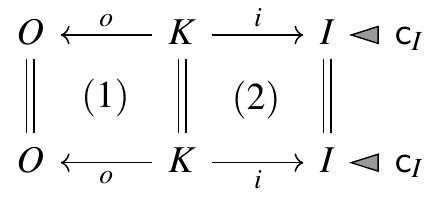}\right)
=\left(\inputtikz{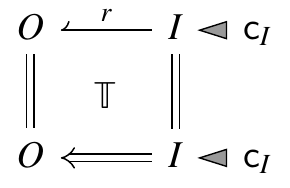}\right)\,.
\end{equation}
\item \emph{Vertical pasting:} for $\bT\in\{DPO,DPO^{\dag},SqPO\}$, and suppose that in the diagrams below the monomorphism $(X'\hookleftarrow X)\circ(X\hookleftarrow I)$ satisfies the condition $\ac{c}_I$ (not explicitly drawn for clarity). Then the following properties hold true: composing squares of the underlying commutative diagrams vertically (indicated by the $\rightsquigarrow$ notation), one obtains (a) for all combinations of types of rewriting with $\bT=\bT'$, or (b) for $\bT$ arbitrary and $\bT'=DPO$ or $\bT'=DPO^{\dag}$, the following $\bT$-type direct derivations:
\begin{equation}\label{eq:defPasting}
\inputtikz{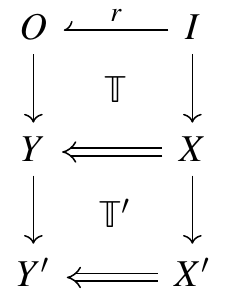}\quad \rightsquigarrow\quad
\inputtikz{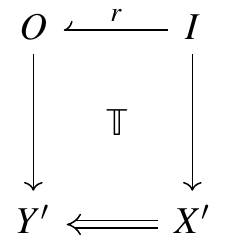}
\end{equation}
\end{romanenumerate}
\end{lemma}
\begin{proof}
  The first statement follows since the squares marked $(1)$ and $(2)$ are pushouts for arbitrary morphisms $o$ and $i$. The second property follows by invoking various elementary square composition lemmata for pushouts and FPCs  (see e.g.\ the list of technical Lemmas provided in the appendix of~\cite{nbSqPO2019}).
\end{proof}

\subsection{Compositional concurrency}\label{sec:concur}

A central role in our rewriting frameworks is played by the notion of certain concurrency theorems, which entail an equivalence between (i) sequential applications of rewriting rules along admissible matches and (ii) application of \emph{composites} of rewriting rules along admissible matches. This structure is in turn intimately related to notions of \emph{traces} and analyses thereof in rewriting theory. In the form as presented below (which is compatible with the notion of compositionality of rewriting rules), both the DPO- and the SqPO-type concurrency theorems were first introduced in~\cite{bk2019a} (with some earlier results in the DPO-setting under weaker assumptions reported in~\cite{habel2009correctness,ehrig2014mathcal}).

\begin{theorem}[Concurrency theorems~\cite{bp2018,nbSqPO2019,bk2019a}]\label{thm:concur}
	Let $\bfC$ be a category satisfying Assumption~\ref{as:main}, and let $\bT\in\{DPO,SqPO\}$. Then there exists a bijection $\varphi:A\xrightarrow{\cong}B$ on pairs of $\bT$-admissible matches between the sets $A$ and $B$,
	\begin{equation}
		\begin{aligned}
			A&=\{(m_2,m_1)\mid m_1\in \MatchGT{\bT}{R_1}{X_0}\,,; 
			m_2\in \MatchGT{\bT}{R_2}{X_1}\}\\
			\cong\quad 
			B&=\{(\mu_{21},m_{21})\mid \mu_{21}\in \MatchGT{\bT}{R_2}{R_1}\,,\; m_{21}\in \MatchGT{\bT}{R_{21}}{X_0}\}\,,
		\end{aligned}
		\end{equation}
	where $X_1=R_{1_{m_1}}(X_0)$ and $R_{21}=\compGT{\bT}{R_2}{\mu_{21}}{R_1}$ such that for each corresponding pair $(m_2,m_1)\in A$ and $(\mu_{21},m_{21})\in B$, it holds that
		\begin{equation}
			R_{21_{m_{21}}}(X_0) \cong
			R_{2_{m_2}}(R_{1_{m_1}}(X_0))\,.
		\end{equation}
\end{theorem}
\begin{proof}
	See~\cite{bk2019a} for the full technical details. Let us note here that the bijective correspondence is \emph{constructive}, i.e.\ there exist explicit algorithms for realizing $B$ from $A$ and vice versa.
\end{proof}

The following technical result is a necessary prerequisite for deriving the proof of Theorem~\ref{thm:tlChar}:

\begin{corollary}\label{cor:concur}
  Let $r_{n\cdots 1}=(O_{n\cdots 1}\Leftarrow I_{n\cdots 1})$ be a span of monomorphisms, and let $(Y^{(n)}_{j+1,j}\Leftarrow Y^{(n)}_{j,j-1})$ be $n$ spans of monomorphisms (for $0\leq j\leq n$) with $Y^{(n)}_{n+1,n}=O_{n\cdots 1}$, $Y^{(n)}_{1,0}=I_{n\cdots 1}$, and such that
  \[
    (O_{n\cdots 1}\Leftarrow I_{n\cdots 1})=
    (O_{n\cdots 1}\Leftarrow Y^{(n)}_{n,n-1})\circ \dotsc \circ
    (Y^{(n)}_{2,1}\Leftarrow I_{n\cdots 1})\,.
  \]
  Let $\ac{c}_{I_{n\cdots 1}}$ be a condition over $I_{n\cdots 1}$. Then for each object $X_0$ and for each $\bT$-admissible match $(I_{n\cdots 1}\hookrightarrow X_0)\in \MatchGT{\bT}{R_{n\cdots 1}}{X_0}$ (for $R_{n\cdots 1}=(r_{n\cdots 1},\ac{c}_{I_{n\cdots 1}})$ and for $\bT\in\{DPO,SqPO\}$), the $\bT$-type application of $R_{n\cdots 1}$ to $X_0$ along this match
\begin{subequations}\label{eq:concurCor}
  \begin{equation}\label{eq:concurCorA}
  \inputtikz{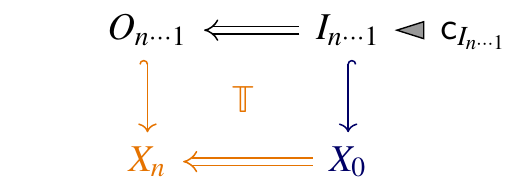}
    \end{equation}
   uniquely (up to isomorphisms) encodes an $n$-step $\bT$-type derivation sequence of the following form, and vice versa:
  \begin{equation}\label{eq:concurCorB}
  \inputtikz{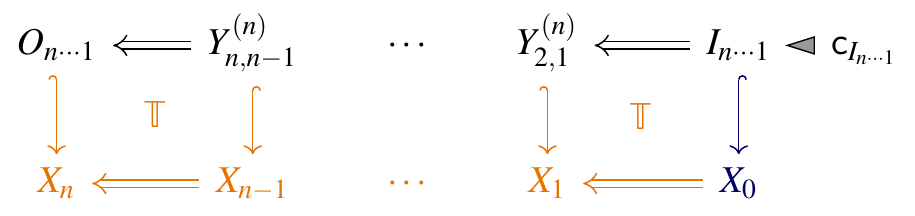}
    \end{equation}
  \end{subequations}
  \begin{proof}
    The statement is trivially true for $n=1$. For $n=2$, note first that for any two composable spans of monomorphisms $S_2=(Z\Leftarrow Y)$ and $S_1=(Y\Leftarrow X)$, invoking the definition of rule compositions (Definition~\ref{def:Rcomp}) allows to verify that considering $S_2$ and $S_1$ as linear rules without conditions, and letting $\mu=(Y\hookleftarrow Y\hookrightarrow Y)$ be a span of identity morphisms on $Y$, the span composition $S_2\circ S_1$ is in fact computable as a rule composition:
    \begin{equation}\label{eq:defRcompAux}
\inputtikz{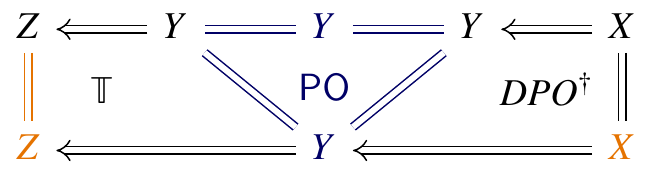}
  \;\rightsquigarrow\quad  S_2\circ S_1=\compGT{\bT}{S_2}{\mu}{S_1}\,.
\end{equation}
Here, according to Lemma~\ref{lem:sqCompAux}\ref{lem:Auxa} the $\bT$- and $DPO^{\dag}$-type direct derivation subdiagrams as indicated always exist. Consequently, the claim of the corollary for $n=2$ follows by invoking the concurrency theorem (Theorem~\ref{thm:concur}) for the special case of $r_{2\cdots 1}=(O_{2\cdots 1}\Leftarrow I_{2\cdots 1})$, $(O_{2\cdots 1}\Leftarrow I_{2\cdots 1})=
    (O_{2\cdots 1}\Leftarrow Y^{(2)}_{2,1})\circ
    (Y^{(2)}_{2,1}\Leftarrow I_{2\cdots 1})$ and for $\mu_{21}$ a span of identity morphisms of $Y^{(2)}_{2,1}$. The proof for the case $n\geq 2$ may then be obtained via induction over $n$.
  \end{proof}
\end{corollary}

\subsection{Compositional associativity}

The second main ingredient of our novel \emph{compositional} rewriting frameworks is the notion of \emph{associativity}. Intuitively, if one wishes to extend the analysis of traces as suggested via the concurrency theorems to a full-fledged analysis that is centered on compositions of rewriting rules (which constitutes the main contribution of this paper in the form of the tracelet framework), one must necessarily have a certain property fulfilled, in that multiple sequential compositions of rewriting rules may be computed in any admissible order of pairwise compositions. The latter feature is crucial for the purposes of analysis of classes of traces, since the traditional interpretation of the concurrency theorem would only permit to reason on pairwise sequential compositions (but not on extension thereof to higher order composites).

\begin{theorem}[Associativity of rule compositions~\cite{bp2018,nbSqPO2019,bk2019a}]\label{thm:assocR}
	Let $\bfC$ be a category satisfying Assumption~\ref{as:main}. let $R_1,R_2,R_3\in \LinAc{\bfC}$ be linear rules with conditions, and let $\bT\in\{DPO,SqPO\}$. Then there exists a bijection $\varphi:A\xrightarrow{\cong} B$ of sets of pairs of $\bT$-admissible matches $A$ and $B$, defined as
	\begin{equation}
		\begin{aligned}
			A&:=\{(\mu_{21},\mu_{3(21)})\mid \mu_{21}\in 
			\MatchGT{\bT}{R_2}{R_1}\,,\; \mu_{3(21)}\in \MatchGT{\bT}{R_3}{R_{21}}\}\qquad (R_{21}=\compGT{\bT}{R_2}{\mu_{21}}{R_1})\\
			B&:=\{(\mu_{32},\mu_{(32)1})\mid 
			\mu_{32}\in\MatchGT{\bT}{R_3}{R_2}\,,\;
			 \mu_{(32)1}\in \MatchGT{\bT}{R_{32}}{R_1}\}
			 \qquad (R_{32}=\compGT{\bT}{R_3}{\mu_{32}}{R_2})
		\end{aligned}
		\end{equation}
		such that for each corresponding pair $(\mu_{21},\mu_{3(21)})\in A$ and $\varphi(\mu_{21},\mu_{3(21)})=(\mu_{32}',\mu_{(32)1}')\in B$, 
		\begin{equation}
			\compGT{\bT}{R_3}{\mu_{3(21)}}{\left(\compGT{\bT}{R_2}{\mu_{21}}{R_1}\right)}\cong
			\compGT{\bT}{\left(\compGT{\bT}{R_3}{\mu_{32}'}{R_2}\right)}{\mu_{(32)1}'}{R_1}\,.
		\end{equation}
	In this particular sense, the composition operations $\compGT{\bT}{.}{.}{.}$ are \textbf{associative}.
\end{theorem}

\subsection{Proof of Theorem~\ref{thm:TmainA}}\label{app:TmainAproof}

\paragraph{Ad part~\ref{thm:TmainApartI}:} Note first that by virtue of Corollary~\ref{cor:concur}, the composition of two $\bT$-type tracelets $T',T\in \cT^{\bT}$ of lengths $m$ and $n$, respectively, along a $\bT$-admissible match $\mu=(I_{m\cdots 1}'\hookleftarrow {\color{h1color}M}\hookrightarrow O_{n\cdots 1})\in \tMatchGT{T'}{T}{\bT}$ encodes a $\bT$-type composition of linear rules $R'_{m\cdots 1}=(r_{m\cdots 1}',\ac{c}_{I_{m\cdots 1}'})$ and $R_{n\cdots 1}=(r_{n\cdots 1},\ac{c}_{I_{n\cdots 1}})$ (with $r_{m\cdots 1}'=(O_{m\cdots 1}'\Leftarrow I_{m\cdots 1}')$ and $r_{n\cdots 1}=(O_{n\cdots 1}\Leftarrow I_{n\cdots 1})$) of the following form:
\begin{equation}\label{eq:proofThmMainA}
\inputtikz{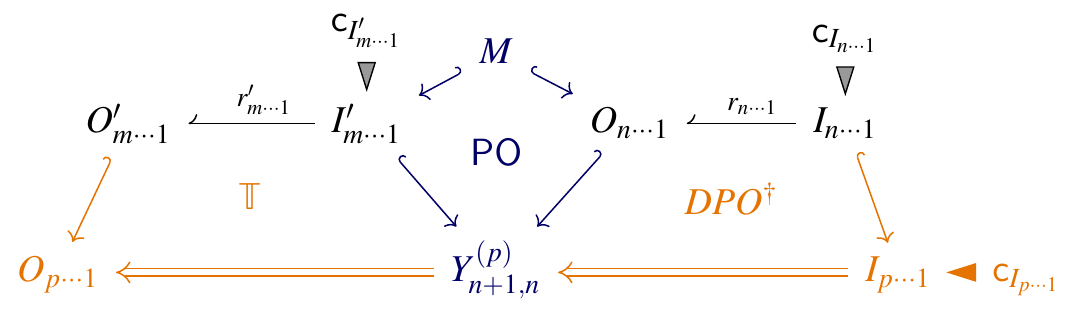}
\end{equation}
Consequently, the constructability of this diagram indeed hinges on whether or not $\mu$ is a $\bT$-type admissible match of $r_{m\cdots1}'$ into $r_{n\cdots 1}$, thus proving the statement of part $\ref{thm:TmainApartI}$.

\paragraph{Ad part~\ref{thm:TmainAii}:} The latter argument has the additional consequence that for all tracelets $T',T\in\cT^{\bT}$ and $\mu\in \tMatchGT{T'}{T}{\bT}$, equation~\eqref{eq:proofThmMainA} demonstrates explicitly that $\left[\left[\TcompGT{T'}{\mu}{T}{\bT}\right]\right]
  =\compGT{\bT}{[[T']]}{\mu}{[[T]]}$, thus proving part~\ref{thm:TmainAii} of the theorem statement.

\paragraph{Ad part~\ref{thm:TmainAiii}:} The proof of part~\ref{thm:TmainAiii} of Theorem~\ref{thm:TmainA} follows by combining the statements of the first two parts with the associativity theorem for $\bT$-type rule compositions (Theorem~\ref{thm:assocR}).

\subsection{Proof of Theorem~\ref{thm:tlChar}}\label{app:tlChar}

The first part of the claim follows by applying a corollary of the concurrency theorem for rules with conditions (Corollary~\ref{cor:concur} of Appendix~\ref{sec:concur}) in order to construct the lower row in the left diagram of~\eqref{eq:tlCharAux}, followed by vertically composing squares (Lemma~\ref{lem:sqCompAux} of Appendix~\ref{app:aux}) in each column of the diagram in order to obtain the derivation trace shown on the right of~\eqref{eq:tlCharAux}. The second part of the statement follows by an inductive application of the concurrency theorem: the case $n=1$ coincides with the definition of a direct derivation, while for $n=2$ Theorem~\ref{thm:concur} precisely describes the transition from a length $2$ derivation trace to a length $1$ derivation trace along the composite rule. The induction step $n\to n+1$ is then verified by applying the concurrency theorem to the derivation trace $X_{n+1}\Leftarrow X_n\Leftarrow X_0$ along the rules $r_{n+1}$ and $(O_{n\dotsc 1}\Leftarrow I_{n\dotsc 1})$.

\section{Compositional sequential independence}\label{app:CSI}

A key role in the analysis of rewriting theories is played by the notion of sequential independence, which we first recall in its traditional form as known from the rewriting literature:

\begin{definition}[cf.\ e.g.\ \cite{ehrig2014mathcal}, Def.~4.3 (DPO) and~\cite{BALDAN2014}, Def.~2.15 (SqPO)]\label{def:tradSI}
Consider a two-step sequence of rule applications of type $\bT$ to an initial object $X_0\in \bfC$ along admissible matches,
\begin{equation}
\inputtikz{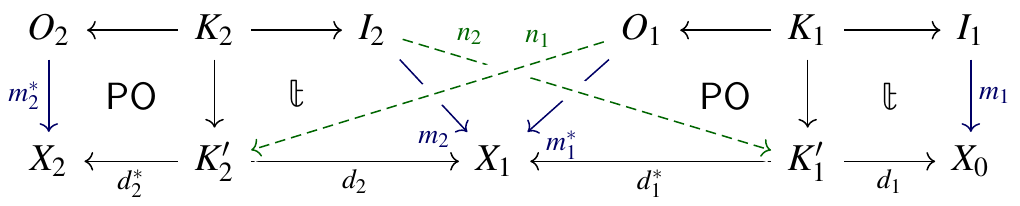}\,,
\end{equation}
where the squares marked $\mathbb{t}$ for $\bT=DPO$ are pushout complements and for $\bT=SqPO$ FPCs. The two rule applications are called \emph{sequentially independent} if there exist monomorphisms $(n_1:O_1\rightarrow K_2'),(n_2:I_2\rightarrow K_1')\in\mono{\bfC}$ such that
\begin{equation}
  d_2\circ {\color{h4color}n_1}={\color{h1color}m_1^{*}}
  \quad \land\quad
  d_1^{*}\circ {\color{h4color}n_2}={\color{h1color}m_2}\,.
\end{equation}
\end{definition}

Based upon the concurrency theorems for the DPO- and SqPO-type compositional rewriting theories, one may develop the following refined variant of the above definition, as was  anticipated e.g.\ in~\cite{Boehm1987aa} for the DPO-type setting:

\begin{lemma}[Compositional sequential independence]\label{lem:CI}
  In the setting of Def.~\ref{def:tradSI}, the rule applications are sequentially independent if and only if there exist monomorphisms $(O_1\rightarrow \overline{K}_2),(I_2\rightarrow\overline{K}_1)\in \mono{\bfC}$ such that
  \begin{equation}
    (N_{21}\leftarrow \overline{K}_2)\circ (\overline{K}_2\leftarrow O_1)=(N_{21}\leftarrow O_1)\quad \land
    \quad
    (N_{21}\leftarrow \overline{K}_1)\circ (\overline{K}_1\leftarrow I_2)=(N_{21}\leftarrow I_2)\,,
  \end{equation}
  with notations as in the explicit version of diagram~\eqref{eq:defRcomp} (see the proof). For the case of rules with conditions, it is in addition required that $\ac{c}_{I_{21}}\,\dot{\equiv}\, \ac{c}_{I_{12}}$.
\end{lemma}
\begin{proof}
  For the ``$\Rightarrow$'' direction, suppose the $\bT$-type rule applications are sequentially independent. Invoking the $\bT$-type concurrency theorem, we may construct the following diagram:
\begin{equation}\label{eq:RcompCCaux}
\inputtikz{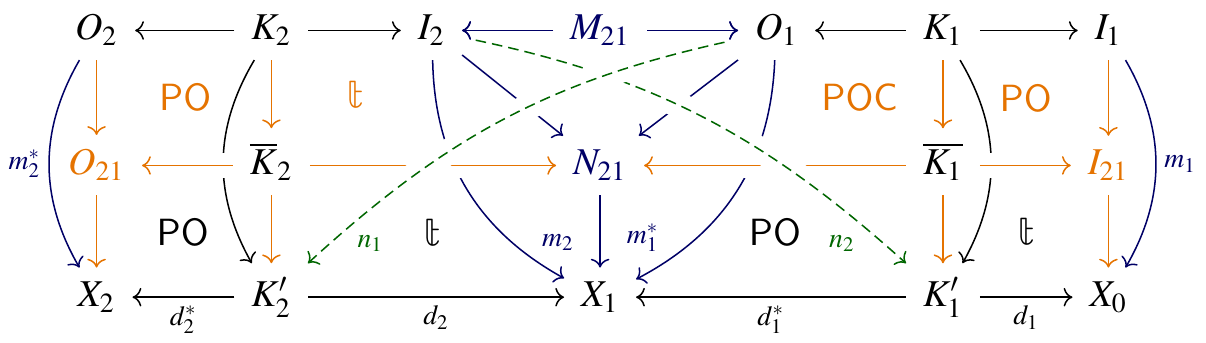}
\end{equation}
Here, all arrows are monomorphisms, the squares marked $\mathbb{t}$ are pushout complements in the DPO- and FPCs in the SqPO-type rewriting case, while ${\color{h1color}M_{21}}=\pB{I_2\rightarrow X_1\leftarrow O_1}$, and ${\color{h1color}N_{21}}=\pO{I_2 {\color{h1color}\leftarrow M_{21}\rightarrow } O_1}$. Since pushouts along monomorphisms are pullbacks in an adhesive category, existence of the morphisms $(I_2{\color{h1color}\rightarrow N_{21}})$ and $(I_2{\color{h4color}\rightarrow}K_1')$ (the latter by assumption) entails the existence of a morphism $(I_2\rightarrow \overline{K}_1)$ such that
\[
  (K_1'\leftarrow \overline{K_1})\circ (\overline{K}_1\leftarrow I_2)=(K_1'{\color{h4color}\leftarrow}I_2)\,.
\]
By the decomposition property of monomorphisms in an adhesive category, $(I_2\rightarrow \overline{K}_1)$ is a monomorphism, too. Analogously, since POCs and FPCs along monomorphisms are also pullbacks, we may infer the existence of a monomorphism $(O_1\rightarrow \overline{K}_2)$ such that
\[
  (K_2'\leftarrow \overline{K_2})\circ (\overline{K}_2\leftarrow O_1)=(K_2'{\color{h4color}\leftarrow}O_1)\,.
\]

The statement of the ``$\Leftarrow$'' direction follows by composition of the relevant monomorphisms (to obtain the monomorphisms $(I_2\rightarrow \overline{K}_2)$ and $(O_1\rightarrow \overline{K}_2)$) and by composing PO and FPC squares (first row with the second row) to obtain the claim. 

Finally, the requirement on the conditions as stated arises from the definition of rule compositions, and is necessary so that the composites of the rules in the two sequential orders can give rise to an $\bT$-admissible match (in the sense of rules with conditions) $I_{21}\hookrightarrow X_0$. 
\end{proof}

The latter result clarifies the precise relationship between compositions and sequential applications of rules on the one hand and the notion of sequential commutativity on the other hand: two rules in a sequential rule application are independent if and only if their underlying concurrent rule composition satisfies a certain property as described above. In other words, we obtain a sharper notion of sequential independence in the latter, compositional form, since this notion characterizes sequential commutativity for an entire class of sequential rule applications (i.e.\ for all $\bT$-type sequential applications that are equivalent to applications of the $\bT$-composite rule $\compGT{\bT}{r_2}{\mu_{21}}{r_1}$). This permits us to provide a refinement of the notion of \emph{switching couples} (cf.\ e.g.~\cite{BALDAN2014}) to the level of rule compositions for rules with conditions, which is the first result of this kind for both DPO- and SqPO-rewriting:

\begin{theorem}[Compositional concurrent commutativity]\label{thm:CCC}
  Let $\bfC$ be a category satisfying Assumption~\ref{as:main}, let $R_j=(r_j,\ac{c}_{I_j})\in\LinAc{\bfC}$ be two rules with conditions, and let $\mu_{21}=(I_2\leftarrow M_{21}\rightarrow O_1)$ be a $\bT$-admissible match of $R_2$ into $R_1$ (with $\bT\in\{DPO,SqPO\}$). Then if there exist monomorphisms $(O_1\rightarrow \overline{K}_2),(I_2\rightarrow \overline{K}_1)\in\mono{\bfC}$ (with notations as in~\ref{lem:CI}), with
  \begin{equation}
    (N_{21}\leftarrow \overline{K}_2)\circ (\overline{K}_2\leftarrow O_1)=(N_{21}\leftarrow O_1)\quad \land
    \quad
    (N_{21}\leftarrow \overline{K}_1)\circ (\overline{K}_1\leftarrow I_2)=(N_{21}\leftarrow I_2)\,,
  \end{equation}
  the following statements hold:
  \begin{romanenumerate}
    \item The pullbacks
    \begin{equation}
    \inputtikz{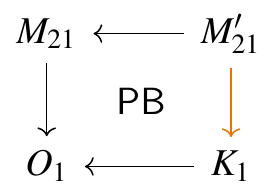}\qquad
      \inputtikz{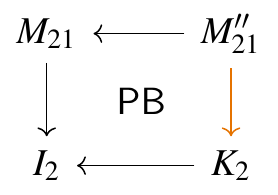}
    \end{equation}
    satisfy $M_{21}'\cong M_{21}$ and $M_{21}''\cong M_{21}$, thus furnishing monomorphisms $(a_1:M_{21}\rightarrow K_1)$ and $(a_2:M_{21}\rightarrow K_2)$.
    \item The span $\mu_{12}:=(I_1\xleftarrow{i_1\circ a_1}M_{21}\xrightarrow{o_2\circ a_2}O_2)$ is a DPO- (and thus SqPO-) admissible match of $r_1$ into $r_2$.
    \item $r_1$ and $r_2$ are sequentially independent w.r.t.\ $\mu_{12}$ in both types of rewriting, and with
    \begin{equation}
      \left(\compGT{SqPO}{r_2}{\mu_{21}}{r_1}\right)
      \cong \left(\compGT{DPO}{r_2}{\mu_{21}}{r_1}\right)
      \cong \left(\compGT{DPO}{r_1}{\mu_{12}}{r_2}\right)
      \cong \left(\compGT{SqPO}{r_1}{\mu_{12}}{r_2}\right)\,.
    \end{equation}
    \item $R_1$ and $R_2$ are sequentially independent only if in addition $\ac{c}_{I_{21}}\equiv \ac{c}_{I_{12}}$.
  \end{romanenumerate}
\end{theorem}
\begin{proof}
Let $K_{21}$ denote the pullback of $\overline{K}_2\hookrightarrow N_{21}\hookleftarrow\overline{K}_1$, and let the square marked $t$ in~\eqref{eq:seqIndProof} below be a pushout complement (POC) for $\bT=DPO$, and a final pullback complement (FPC) for $\bT=SqPO$, respectively. 

\paragraph{Ad~(i) and~(ii):} By assumption, there exist monomorphisms $O_1\hookrightarrow \overline{K}_2$ and $I_2\hookrightarrow \overline{K}_1$. This entails by virtue of the universal property of pullbacks the following structures (compare~\eqref{eq:seqIndProof}):
\begin{itemize}
\item Commutativity of the diagram and existence of the monomorphisms $\overline{K}_1\hookleftarrow K_1$ and $(\overline{K_2}\hookleftarrow O_1)\circ (O_1\hookleftarrow K_1)$ entail by the universal property of the pullback $\square(K_{21},\overline{K}_1,N_{21},\overline{K}_2)$ the existence of a morphism $K_{21}\leftarrow K_1$. Since $(\overline{K}_1\hookleftarrow K_1)=(\overline{K}_{21}\hookleftarrow K_{21})\circ (K_{21}\leftarrow K_1)$, with $\hookleftarrow$ indicating monomorphisms, by stability of monomorphisms under decompositions we conclude that $(K_{21}\leftarrow K_1)$ is a monomorphism. In an analogous fashion, the existence of monomorphisms $K_2\hookrightarrow \overline{K}_2$ and $M_{21}\hookrightarrow K_i$ (for $i=1,2$) may be established. This proves part (i).
\item Invoking \emph{pushout-pullback decomposition} (cf.\ e.g.\ \cite{nbSqPO2019}) repeatedly (noting that POCs and FPCs along monomorphisms are also pullbacks), we may conclude that the squares marked $t$ and $PB$ in~\eqref{eq:seqIndProof} are pushouts, as are all squares formed involving the monomorphisms previously discussed to exist (i.e.\ the morphisms marked in light blue in~\eqref{eq:seqIndProof}). 
\end{itemize}
\begin{equation}\label{eq:seqIndProof}
\vcenter{\hbox{\includegraphics[scale=0.5,page=1]{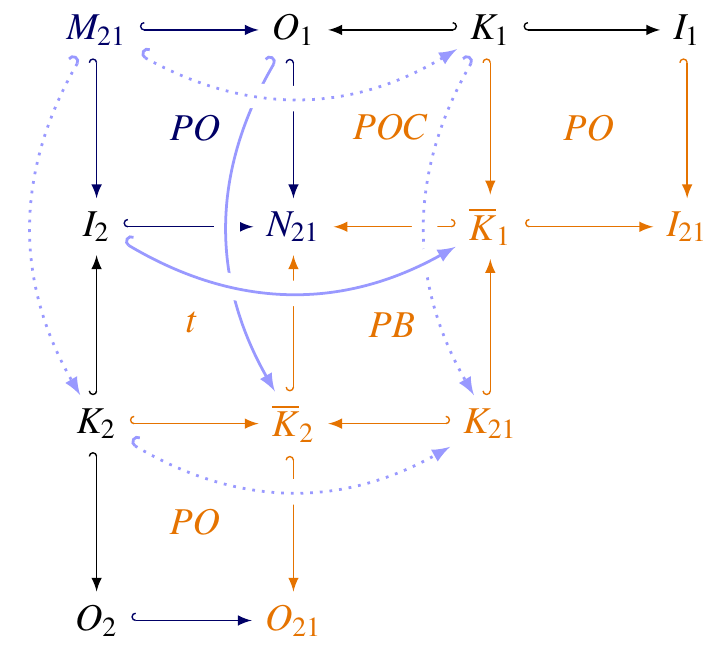}}}\;
\rightsquigarrow\;
\vcenter{\hbox{\includegraphics[scale=0.5,page=2]{images/seqIndAux.pdf}}}\;
\rightsquigarrow\;
\vcenter{\hbox{\includegraphics[scale=0.5,page=3]{images/seqIndAux.pdf}}}
\end{equation}
From hereon, we may follow the classical strategy for proving sequential commutativity in the DPO-type setting (cf.\ e.g.~\cite{ehrig2006fund}, proof of Thm.~5.12): first, we form the pushouts $\overline{\overline{K}}_2=\pO{K_{21}\hookleftarrow K_1\hookrightarrow I_1}$, $\overline{\overline{K}}_1:=\pO{O_2\hookleftarrow K_2\hookrightarrow K_{21}}$ and $N_{12}:=\pO{\overline{\overline{K}}_1\hookleftarrow K_{21}\hookrightarrow \overline{\overline{K}}_2}$, which by universal properties of pushouts and stability of monomorphisms under decompositions leads to the existence of monomorphisms $\overline{\overline{K}}_2\hookrightarrow I_{21}$ and $\overline{\overline{K}}_1\hookrightarrow O_{21}$. Since moreover by virtue of \emph{pushout-pushout decomposition} the newly formed squares involving the two aforementioned monomorphisms are found to be pushouts, we finally obtain a DPO-type composition of $r_1$ with $r_2$ along the span $\mu_{21}$ by assembling pushout squares as depicted in the last step of~\eqref{eq:seqIndProof}. This identifies the span $\mu_{12}:=(I_1\hookleftarrow M_{21}\hookrightarrow O_2)$ as a DPO- (and thus SqPO-) admissible match of $r_1$ into $r_2$, which proves part (ii). 

\paragraph{Ad~(iii):} Since we have found in the proof of part~(i) that the square marked $t$ in~\eqref{eq:seqIndProof} is a pushout whenever the monomorphisms $O_1\hookrightarrow \overline{K}_2$ and $I_2\hookrightarrow \overline{K}_1$ exist, note first that sequential compositions of rules $r_2$ and $r_1$ along a $\bT$-admissible match $\mu_{21}$ that are sequentially independent are in fact always DPO-type compositions (which for $\bT=SqPO$ is indeed a possible special case, since a pushout complement is also an FPC). Together with the construction of the DPO-type composition of $r_1$ and $r_2$ along the uniquely induced span $\mu_{12}$ as presented in the proof of part~(ii), which in particular entails the existence of monomorphisms $O_2\hookrightarrow \overline{K}_1$ and $I_1\hookrightarrow \overline{K}_2$, this provides the proof of part~(iii).

\paragraph{Ad~(iv):} The final claim follows by verifying the well-known fact that there is no guarantee for the conditions $\ac{c}_{I_{21}}$ and $\ac{c}_{I_{12}}$ of the two composites to coincide, thus concluding the proof.
\end{proof}
Note that the above statements have the peculiar consequence that two sequentially independent rules $r_2$ and $r_1$ give rise to a so-called \emph{amalgamated rule}~\cite{Boehm1987aa}, in the sense that
\begin{equation}
\begin{aligned}
  O_{21}&=\pO{O_2\hookleftarrow M_{21}\hookrightarrow O_1}\\
  K_{21}&=\pO{K_2\hookleftarrow M_{21}\hookrightarrow K_1}\\
  I_{21}&=\pO{I_2\hookleftarrow M_{21}\hookrightarrow I_1}\,.
\end{aligned}
\end{equation}
Since the theory of amalgamation has been extensively developed in the graph rewriting literature~\cite{Boehm1987aa,Golas2010aa,GOLAS2014,ehrig2014mathcal}, it might well be the case that the above result may be beneficial in the concrete implementations of tracelet analysis algorithms.

\subsection{A worked example of tracelet shift equivalence}\label{app:tlShiftEqExample}

\begin{figure}
    \centering
    \includegraphics[scale=0.3]{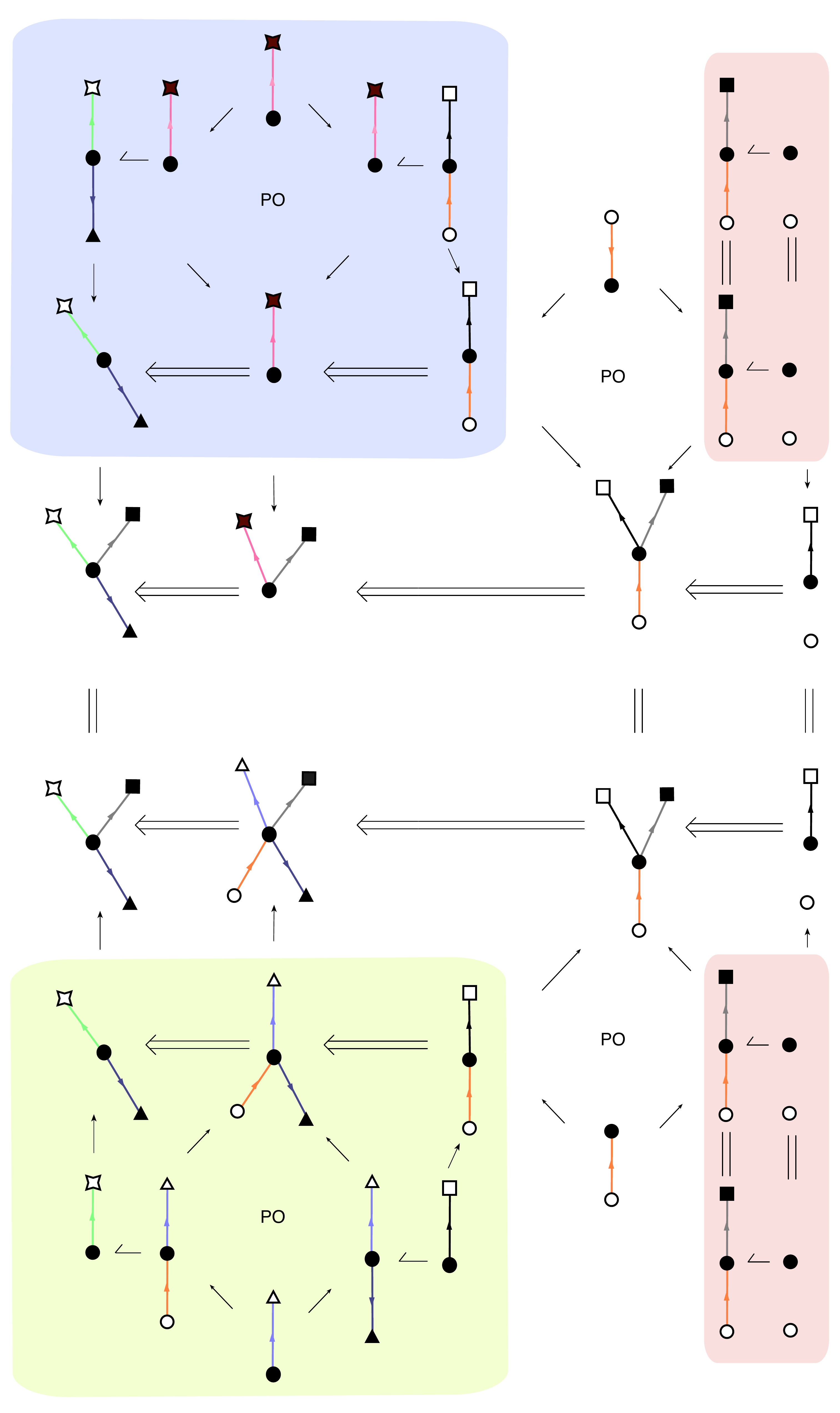}
  \caption{Illustration of a concrete example of tracelet shift equivalence, based upon the tracelet of length $3$ as presented in Figure~\ref{fig:traceletAssociativity}. The bottom half of the diagram is identical to the bottom half of Figure~\ref{fig:traceletAssociativityB}, while the top half illustrates a shift-equivalent tracelet of length $3$ in which the order of the second and third rule applications has been swapped.}\label{fig:traceletShiftEqExample}
\end{figure}

In order to provide some intuitions for the notion of tracelet shift equivalence, we present in Figure~\ref{fig:traceletShiftEqExample} a concrete example of two tracelets of length $3$ that are shift equivalent. The bottom half of the diagram coincides with the bottom half of Figure~\ref{fig:traceletAssociativityB}, while the top half of Figure~\ref{fig:traceletShiftEqExample} encodes a tracelet where the order of applications of the second and third rules has in effect been reversed. Note in particular that while the rules involved are understood as rewriting rules for finite directed (unlabeled) multigraphs, we have employed vertex symbols and edge colors in order to encode the structure of the various monomorphisms and partial maps (i.e.\ repeated symbols encode elements related by partial maps). We have moreover chosen representatives for the two tracelets such that the isomorphisms that relate the tracelets are concretely implemented by isomorphisms of the underlying rewriting rules. An essential feature of our definition of shift equivalence (Definition~\ref{def:tSeq}) is the following technical detail: for tracelets $T=t_n\vert\dotsc\vert t_1$ and $T'=t_n'\vert\dotsc\vert t_1'$ that are shift equivalent based upon subtracelets $t_j\vert\dotsc \vert t_{j-k}$ and $t_j'\vert\dotsc \vert t_{j-k}'$, part of the definition entails that we demand the existence of an isomorphism between $t_n\vert\dotsc \vert t_{(j\vert\dotsc \vert j-k)}\vert \dotsc \vert t_1$ and $t_n'\vert\dotsc \vert t_{(j\vert\dotsc \vert j-k)}'\vert \dotsc \vert t_1'$. However, we do \emph{not} demand an isomorphism between the original tracelets $T$ and $T'$, which would only exist in the special case where the subtracelets encode sequentially independent derivations in the traditional sense. This feature is illustrated explicitly in Figure~\ref{fig:traceletShiftEqExample}, where the minimal derivation traces encoded by the two tracelets of length $3$ are in fact \emph{not} in isomorphism (due to the non-existence of an isomorphism of the ``X-shaped'' respective third objects in the minimal traces that would be compatible with the morphism structure of the traces), but only the minimal derivation traces of the tracelets of length $2$ given by $t_{(3\vert 2)}\vert t_1$ and $t_{(3\vert 2)}'\vert t_1'$, respectively. Here, the tracelet $T_{(3\vert 2)}'$ of length $2$ that leads to $t_{(3\vert 2)}'$ is depicted in the light blue box, while $T_{(3\vert 2)}$ leading to $t_{(3\vert 2)}$ is depicted in the light yellow box. It is this particular feature that deserves to refer to the process of abstracting tracelets by means of tracelet shift equivalence as a form of \emph{strong compression} in the sense of~\cite{danos2012graphs}.

\end{document}